\documentclass[sigconf]{acmart}

\usepackage{booktabs} % For formal tables
\usepackage{amsfonts}
\usepackage{amssymb,amsmath,amsthm}
\usepackage{graphicx}
\usepackage{environ}
\usepackage{color}
\usepackage{epsfig}
\usepackage{caption}
\usepackage{subfigure}
\usepackage{verbatim}
\usepackage{float}
\usepackage{calc}
\usepackage{multirow}
\usepackage{hhline}
\usepackage{enumerate}

\def\conference{0}

\newtheorem{theorem}{Theorem}
\newtheorem{corollary}[theorem]{Corollary}

\newtheorem{algorithm}{Algorithm}

\newenvironment{labellist}[1][A]
{\begin{list}{{#1}\arabic{enumi}.}{\usecounter{enumi}\addtolength{\leftmargin}{-1ex}
      \addtolength{\labelwidth}{\widthof{{#1}5}}}} 
{\end{list}}

\newcommand{\LAT}{\operatorname{Lat}}

\renewcommand{\min}{{\operatorname{min}}}
\renewcommand{\max}{{\operatorname{max}}}

\DeclareMathOperator{\E}{E}
\DeclareMathOperator{\poly}{poly}

\DeclareMathOperator{\conv}{conv}

\newcommand{\R}{\ensuremath{\mathbb R}}

\newcommand{\I}{\ensuremath{\mathcal I}}

\newcommand{\Qc}{\ensuremath{\mathcal Q}}

\newcommand{\OPT}{\ensuremath{\mathit{OPT}}}

\newcommand{\sm}{\ensuremath{\setminus}}
\newcommand{\es}{\ensuremath{\emptyset}}
\newcommand{\sse}{\subseteq}

\newcommand{\cg}{\ensuremath{\mathit{CG}}}

\newcommand{\lb}{\ensuremath{\mathsf{LB}}}

\newcommand{\into}{\ensuremath{\mathrm{in}}}
\newcommand{\out}{\ensuremath{\mathrm{out}}}

\newcommand\Time{\ensuremath{\mathsf{T}}}

\newcommand{\gm}{\ensuremath{\gamma}}

\newcommand{\ld}{\ensuremath{\lambda}}

\newcommand{\dt}{\ensuremath{\delta}}

\newcommand{\assign}{\ensuremath{\leftarrow}}

\newcommand{\inpsize}{\ensuremath{\I}}

\setcopyright{rightsretained}
\begin{document}
\title[Minimizing Latency in Online Services]{Minimizing Latency in Online Ride and Delivery Services}
\if\conference1
\titlenote{The full version of the paper is available at TODO.}
\else
\titlenote{A short version of the paper is to appear at the 27th Web Conference (formerly, World Wide Web Conference), 2018.}
\fi
%\subtitle{Extended Abstract}

\author{Abhimanyu Das}
\affiliation{%
  \institution{Google Research}
}
\email{abhidas@google.com}

\author{Sreenivas Gollapudi}
\affiliation{%
  \institution{Google Research}
}
\email{sgollapu@google.com}

\author{Anthony Kim}
\authornote{Part of this work was done when the author was an intern at Google, Inc.}
\affiliation{%
  \institution{Stanford University}
}
\email{tonyekim@stanford.edu}

\author{Debmalya Panigrahi}
\affiliation{%
  \institution{Duke University}
}
\email{debmalya@cs.duke.edu}

\author{Chaitanya Swamy}
\affiliation{%
  \institution{University of Waterloo}
}
\email{cswamy@uwaterloo.ca}

\renewcommand{\shortauthors}{Das et al.}

\begin{abstract}
Motivated by the popularity of online ride and delivery services, we study natural variants of classical multi-vehicle minimum latency problems where the objective is to route a set of vehicles located at depots to serve request located on a metric space so as to minimize the total latency. In this paper, we consider point-to-point requests that come with source-destination pairs and release-time constraints that restrict when each request can be served. The point-to-point requests and release-time constraints model taxi rides and deliveries. For all the variants considered, we show constant-factor approximation algorithms based on a linear programming framework. To the best of our knowledge, these are the first set of results for the aforementioned variants of the minimum latency problems. Furthermore, we provide an empirical study of heuristics based on our theoretical algorithms on a real data set of taxi rides.
\end{abstract}




\maketitle

\section{Introduction} \label{sec:intro}

In recent years, ride-sharing platforms such as Lyft, Ola and Uber, and online delivery services such as DoorDash and GrubHub have become increasingly popular and have expanded their operations to many cities and countries. A central problem common to these online services is the vehicle routing problem where a fleet of vehicles are routed to serve ride and delivery requests over a geographical area. Indeed, this problem is also at the core of traditional city taxi services, such as Yellow Cab, where taxis are routed to serve ride requests received over the phone or the Internet. In all these settings, a request comprises a pair of source and destination locations, such as the rider's starting and ending coordinates for taxi services, and the restaurant and customer addresses for food delivery services. Furthermore, these ``point-to-point'' requests typically have release-time constraints, i.e., the customer specifies a desired service time before which the request cannot be serviced. The vehicle routing algorithm desires to minimize the average latency of customers, i.e., the difference between their requested service time and the actual service time. In this paper, we formally define this multi-vehicle routing problem, obtain an algorithm based on a linear programming framework with a formal guarantee on its performance, and demonstrate that heuristics based on this formal algorithm improve on benchmark greedy solutions on real data sets of city taxi rides.
Most current systems, such as the online\footnote{Note that we use ``online'' to refer to the requests being generated by web-based services, and not to refer to requests arriving ``online'' in an algorithmic sense.}  and traditional services mentioned above, employ various heuristics for a batch of requests to solve this kind of vehicle routing problems in practice. In contrast, there is a rich history in the algorithmic literature on the so-called {\em vehicle routing problems} (or VRP), which covers a wide range of routing problems for one or more vehicles under a variety of constraints. For these problems, the formal literature comprises a wide array of sophisticated techniques, often based on linear programing formulations, that lead to approximation algorithms with provable guarantees. While we are not aware of any previous results on our exact problem formulation with point-to-point requests and release-time constraints, the related literature raises the natural question: can these sophisticated algorithmic techniques be brought to bear on this important practical problem of minimizing latency for point-to-point requests with release-time constraints? And if so, do these algorithmic ideas also lead to better heuristics in practice? We answer both these questions affirmatively by designing a constant approximation algorithm for this problem, which also leads to a heuristic that outperforms a natural greedy strategy.

\subsection{Our Contributions}

We formally define the minimum latency problems in Section~\ref{sec:prob-formulation}. In Section~\ref{sec:lp-framework}, we present the linear programming framework due to \citet{PS15} that will be central in our algorithms and analyses. Our theoretical contributions are to obtain {\bf constant approximation algorithms} for the following problems:
\begin{enumerate}
\item (Section~\ref{sec:single-depot-kmlp}) For the minimum latency problem (MLP) and single-depot, multi-vehicle minimum latency problem ($k$-MLP) with point-to-point requests.%, we obtain polynomial-time constant approximation algorithms. 
\item (Section~\ref{sec:multi-depot-kmlp}) For the multi-depot, multi-vehicle minimum latency problem ($k$-MLP) with point-to-point requests.%, we obtain a polynomial-time 25.488 approximation algorithm. 
\item (Section~\ref{sec:kmlp-release-results}) For the above problems with release-time constraints for both point and point-to-point requests.%, we obtain polynomial-time constant-factor approximation algorithms.
\end{enumerate}

Additionally, we perform a large-scale empirical analysis for the minimum latency problem by comparing our algorithms with a natural greedy baseline using two real-world taxi data sets.  We show that our algorithms outperform the baseline on a set of different metrics including latency, tour length, and utilization (active time) of the cabs in the system.

See Table~\ref{tbl:results} for a summary of our theoretical results with the precise constant approximation factors. 
To the best of our knowledge, {\em these results are the first polynomial-time approximation guarantees for the respective minimum latency problems}.
For multi-depot problems, we obtain approximation guarantees with somewhat large constants via a constant-factor reduction (of ratio 3). We believe better approximation ratios are possible with use of more complicated constructs, but we do not explore these in this paper to keep our algorithms viable in practice.

\smallskip
\noindent
{\em ``Client-side'' vs ``Platform-side'' objectives:}
 Note the average latency objective that we study in this paper is a ``client-side'' objective whereas the total distance traveled by the vehicles, commonly studied in problems such as the well-known {\em traveling salesman problem} (or TSP), is a ``platform-side'' objective. We can interpret the average latency as the average waiting time from the clients' perspective and the total distance as the platform's operation cost for fuel, etc. Our routing problem with point-to-point requests can be thought of as the ``client-side'' counterpart of what is known as the Dial-a-Ride problem with unit-capacity vehicles, where vehicles serve requests with point-to-point requests as in our problem, but seek to minimize the ``platform-side'' objective of total travel distance. The latency minimization problem with point requests subject to release-time constraints has been proposed as an open problem in \cite{T92} and there exist polynomial-time approximation schemes for special cases such as a constant number of vehicles, or if the metric space has a special structure such as the Euclidean plane or weighted trees~\cite{S14}. Our result is the first one for a general metric space and an arbitrary number of vehicles, and generalizes further to point-to-point requests.

%Although not equivalent, the two objectives are closely related and there are corresponding problems and results with respect to both objectives. 

\begin{table*}[!t]
\makebox[\textwidth][c]{
\begin{tabular}{|l|c|c|c|c|}
\hline
	& P Reqs. & P2P Reqs. & P Reqs. w/ RTs & P2P Reqs. w/ RTs \\
\hline
\hline
MLP & 3.592 (\cite{CGRT03}) & 3.592 & 7.183 & 7.183 \\
Single-Depot $k$-MLP & 7.183 (\cite{PS15}) & 8.978 & 7.183 & 8.978 \\
Multi-Depot $k$-MLP & 8.497 (\cite{PS15}) & 25.488 & 13.728 & 41.184 \\
\hline
\end{tabular}}
\caption{The state-of-the-art approximation guarantees for various minimum latency problems (MLP/$k$-MLP) with point (P) requests and point-to-point (P2P) requests and with or without release times (RTs). Except those in the first column, the constant-factor approximation ratios are new and due to this paper.}
\label{tbl:results}
\end{table*}

\subsection{Related Work}

Our problem is an example of a vehicle routing problem, which is a generic term used to describe a wide range of routing problems over metric spaces. Of particular relevance to our context is the {\em minimum latency problem} (or MLP), also known as the traveling repairman problem or the delivery man problem, which has applications in diskhead scheduling and searching information in a network such as the world wide web~\cite{ALM00,KPY96}. This problem is a special case of our problem, where requests are at point locations (instead of point-to-point requests) and there are no release-time constraints. 
MLP and $k$-MLP (respectively, one or $k$ vehicles) have been long studied in both the Operations Research and Computer Science communities. MLP was shown to be NP-complete and then MAXSNP-hard for general metrics~\cite{SG76,BCCPRS94,PY93}. In fact, it is NP-hard even when the metric is a weighted tree with $\{0, 1\}$ weights~\cite{S02} and is thought to be harder than the well-known traveling salesman problem.\footnote{The traveling salesman problem is MAXSNP-hard for general metrics but is solvable in polynomial time in the case of weighted tree metrics (via an Eulerian tour).}

There have been many works focused on exactly solving MLP/$k$-MLP and related problems, albeit not in polynomial time. A number of mixed integer formulations have been proposed and exact methods such as cutting-plane algorithms and branch-cut-and-price algorithms (e.g., \cite{MZL08,AFPU13,LQL14}) and various meta-heuristics (e.g., \cite{MUH13,SSVO12,NMAM16}) have been proposed. More recently, several mixed integer formulations for $k$-MLP have been proposed and experimented with on routing and scheduling instances with the number of nodes ranging up to 80 \cite{ACA17}.

The first constant-factor approximation algorithm for MLP was obtained in \cite{BCCPRS94} and the approximation factor was subsequently improved to $3.592$ in a series of work~\cite{GK96,ALW08,CGRT03}. Similarly, we have constant-factor approximation algorithms for the multi-vehicle version, $k$-MLP, for both the multi-depot and single-depot variants due to a series of work~\cite{CK04,FHR03,CGRT03,PS15}: the current best-known approximation factors are $8.497$ and $7.183$ respectively. For several special cases, stronger guarantees are known. A quasi-polynomial time approximation scheme was known for weighted trees and Euclidean metrics in any finite dimensions \cite{AK03}, and, more recently, a polynomial time approximation scheme was shown for weighted trees and the Euclidean plane for MLP and single-depot $k$-MLP for any constant $k$ \cite{S14}. 

MLP/$k$-MLP are also closely related to other vehicle routing problems such as the traveling salesman problem, orienteering (cf. \cite{CKP12}) and the Dial-a-Ride problem (cf. \cite{DLSSS04}). They are also related to many sequencing problems with minimum total (weighted) completion time objective in the scheduling literature (cf. \cite{GLLR79,S14}). There is a large body of work on vehicle routing problems and scheduling problems beyond the scope of this paper. For further details, we refer to the above work and references therein.

Finally, we mention several work among many on other aspects of the ride and delivery services from the Data Mining and Artificial Intelligence communities. Different taxi dispatching strategies and route-recommendation systems have been studied in order to minimize passengers' waiting times (cf. \cite{AAR09,ZLX12} in different settings from ours), to maximize drivers' profits (cf. \cite{QZLLX14}), and to guarantee fairness within a group of competing drivers (cf. \cite{QCMSL15}). To address inefficiencies in taxi systems, several graph-based models and algorithms have been designed to minimize the total number of required taxis and to reduce the total idle time of taxi drivers (cf. \cite{ZQU14,ZP17}). Dynamic variants where ride requests arrive on demand have also been studied (cf. \cite{SX13}).

\section{Problem Formulation} \label{sec:prob-formulation}

We define the multi-vehicle minimum latency problem ($k$-MLP) with point-to-point requests as follows. Let $G = (V, E)$ be a weighted complete undirected graph with a distance function on edges, $c: E \rightarrow \mathbb R^+$, that forms a metric space. There are $n$ point-to-point requests where each request $R_i$ is given by a pair $(s_i, d_i)$ of source and destination nodes and to be satisfied without interruption, that is, a vehicle serves it by first going to the source node and then to the destination node directly. There are $k$ vehicles located at respective designated depot nodes, equivalently, root nodes, $r_1, \ldots, r_k$. The objective is to find $k$ paths $P_1, \ldots, P_k$ starting from respective depots that serve all the requests and minimize the {\em total latency}, that is, the sum of latencies of requests, where the latency of a request is equal to the distance from the depot to the destination of the request on the path of the vehicle that serves it.\footnote{Note the total latency and average latency differ by the factor of $n$ and optimizing with respect to both objectives are equivalent.}

After appropriate scaling and rounding, we may assume $c_e$ are integers and $c_{r_i s_j} + c_{s_j d_j} \geq 1$ for every root node $r_i$ and request $R_j$. For ease of exposition, we assume the distances $c_e$ are given in a unit of time and interpret the latency of a request to be its completion time, following the job scheduling literature. 

The problem thus defined is the most general version to be studied in this paper and we refer to it as {\em multi-depot $k$-MLP with point-to-point requests}. We refer to the case when there is a single depot $r_0$ for all vehicles (i.e., $r_0 = r_1 = \cdots = r_k$) as {\em single-depot $k$-MLP with point-to-point requests}, and the case with exactly one vehicle (i.e., $k=1$) as {\em MLP with point-to-point requests}. When the requests' source and destination nodes are identical, we have the classical $k$-MLP and we refer to them as {\em (multi-depot/single-depot) $k$-MLP with point requests} (e.g., \cite{PS15}) to distinguish them from the more general problems with point-to-point requests. 

MLP/$k$-MLP {\em with release times} is MLP/$k$-MLP with additional release-time constraints. Each request $R_i$ has release time $T_i$ such that it cannot be visited before time $T_i$. Both point and point-to-point requests can be considered with release times. When a vehicle is at the source node of a request, it may wait there until when the request become available at its release time. When clear from the context, we refer to these problems with shorter names.

\paragraph*{Notations.}
We refer to the designated nodes using the $r_i$, $s_i$ and $d_i$, and to arbitrary nodes of any kind using generic indexing variables such as $u$ and $v$. We use $\inpsize$ to denote the input size. A path may start from and end at different nodes and a tour must start from and end at the same node. We represent a path/tour by the sequence of nodes on the path/tour or simply by the sequence of requests in the order served if it is a vehicle's route. For example, if the order of the requests is $R_1 \cdots R_n$, the corresponding path is $P = r_0 s_1 d_1 \cdots s_n d_n$ (assuming the root $r_0$). For a path $P$, let $\LAT(P, i)$ be the latency of the $i$-th request on the path, i.e., the distance along the path from the root of the path to the destination of the request, and $\LAT(P) = \sum_{i=1}^n \LAT(P, i)$ be the total latency of the path. The total latency objective is the sum of the latencies of the paths of the $k$ vehicles. Note that if there is one vehicle located at the root $r_0$ that serves the all the requests in the order $R_1 \cdots R_n$, then
\[
\LAT(P, i) = \begin{cases} \LAT(P, i-1) + c_{d_{i-1} s_i} + c_{s_i d_i}, & i > 1 \\
	0, & i = 0 \end{cases} \,,
\]
where $d_0 = r_0$.

Given a set of edges $Q$ and a distance function $c$, let $c(Q)$ be the sum of the lengths given by $c$ of the edges in $Q$; if $P$ is a path, then $c(P)$ is the length of the path. We frequently perform concatenation and shortcutting operations on paths and tours. If two paths/tours meet at the same node, that is, one ending at and the other starting from the same node, we may concatenate them back to back to create a longer path/tour. We may shortcut a path/tour to avoid visiting a node twice by skipping it; this leads to a path/tour of at most the original length if the distances satisfy the triangle inequality, which holds in metric spaces.

\section{LP Framework} \label{sec:lp-framework}

We describe the linear programming framework due to \citet{PS15} for single-depot $k$-MLP and MLP with point requests. Some of our algorithms utilize a directed metric, so we describe their LP in this general setting. Linear programs for multi-depot $k$-MLP are also given in \cite{PS15} and are different, but we primarily focus on the linear program for the single-depot case in this paper. Let $r$ denote the single depot for the point-request version.
%We follow the original presentation closely. 

Given a problem instance with point requests, let digraph $D = (V, A)$ with arc-costs $\{c_{u,v}\}_{u,v\in V}$ represent the underlying directed metric. (If we have an undirected metric, we simply bidirect the edges to obtain $D$, setting $c_{u,v}=c_{v,u} = c_{uv}$.) 
%for both $a = (u,v)$ and $a = (v, u)$. 
We use $a$ to index the arcs in $A$, $v$ to index nodes in $V\sm \{r\}$, $i$ to index the $k$ vehicles, and $t$ to index time units in $[\Time]$. We use variables $x^i_{v,t}$ to denote if node $v$ is visited at time $t$ by the route originating at the root. Directing the vehicles' routes away from the root in a solution, $z^i_{a,t}$ indicates if arc $a$ lies on the portion of vehicle $i$'s route up to time $t$. 
%We let $\Time = n \lb$ where $\lb := \max_v (c_{r,v}+c_{v,r})$ which is clearly 
Let $\Time$ be an easily certifiable upper bound on the maximum latency of a request. Consider the following LP. We will be able to ensure that either: (a) $\Time$ is polynomially bounded by scaling and rounding the metric (e.g., \cite{AK03}) while losing a $(1+\epsilon)$-factor, in which case this LP can be solved efficiently; or (b) $\log\Time$ is polynomially bounded, and use ideas from \cite{PS15} to obtain a $(1+\epsilon)$-approximate solution to this LP. 

Constraints \eqref{jasgn3} ensure that every non-root node is visited at some time, and constraints \eqref{rootdist} ensure that each node cannot be visited before the distance from the root is covered. 
Constraints \eqref{jcov3}--\eqref{path3} are for the vehicles' routes: \eqref{jcov3} ensures that the portion of a
vehicle's route up to time $t$ must visit every node visited by that vehicle by time $t$,
\eqref{onep3} ensures that this route indeed has length at most $t$, and finally
\eqref{path3} seeks to encode that the route forms a path. (Note that constraints
\eqref{path3} are clearly valid, and one could also include the constraints
$\sum_{a\in\dt^{\out}(r)}z^i_{a,t}\leq 1$ for all $i$, $t$.) 

\begin{alignat}{3}
\min & \enspace & \sum_{v,t,i} tx^i_{v,t} & \tag{LP} \label{eq:lp} \\
\text{s.t.} && \sum_{t,i} x^i_{v,t} & \ge 1 \quad && \forall v \label{jasgn3} \\ 
&& x^i_{v,t} &=0\ \text{if $c_{r,v}>t$} \quad && \forall v,t,i \label{rootdist} \\
&& \sum_{a\in\dt^{\into}(S)}z^i_{a,t} & \ge \sum_{t'\le t} x^i_{v,t'} 
\quad && \forall S\subseteq V\sm\{r\}, v\in S, \forall t, i \label{jcov3} \\
&& \sum_{a} c_az^i_{a,t} & \le t \quad && \forall t,i \label{onep3} \\
&& \sum_{a\in\dt^{\into}(v)}\negthickspace\negthickspace z^i_{a,t} & \geq\sum_{a\in\dt^{out}(v)}\negthickspace\negthickspace z^i_{a,t} 
\quad && \forall v, t, i \label{path3} \\
&& x, z & \geq 0.
\end{alignat}

To round a fractional solution to a set of routes for vehicles, we use a polynomial-time arborescence packing result for weighted digraphs and the concatenation graph. The following result does not require the edge costs are symmetric or form a metric, but holds for arbitrary nonnegative edge costs.

\begin{theorem}[Theorem 3.1 in \cite{PS15}]\label{thm:arbpoly}
Let $D=(U+r,A)$ be a digraph with nonnegative integer edge weights $\{w_e\}$, 
where $r\notin U$ is a root node and $|\delta^{\into}(u)| \ge |\delta^{\out}(u)|$
for all $u \in U$. 
For any integer $K\geq 0$, one can find out-arborescences $F_1,\ldots,F_q$
rooted at $r$ and integers $\gm_1,\ldots,\gm_q$ in polynomial time such that  
$\sum_{i=1}^q\gm_i=K$, $\sum_{i:e\in F_i}\gm_i\leq w_e$ for all $e\in A$, 
and $\sum_{i:u\in F_i}\gm_i=\min\{K,\ld_D(r,u)\}$ for all $u\in U$.
\end{theorem}

The {\em concatenation graph} was introduced by \citet{GK96} and then extended by \citet{AB10} as a convenient mean of representing the concatenation process of constructing vehicles' routes from shorter paths. The concatenation graph corresponding to a sequence $w_1=0,\ldots,w_n$ of nonnegative numbers (that starts with a $0$), denoted $\cg(w_1,\ldots,w_n)$, is a directed graph with $n$ nodes and an arc $(i,j)$ of length $\bigl(n - \frac{i+j}{2} \bigr) w_j$ for all $i<j$. In our applications, a path through $\cg(w_1, \ldots, w_n)$ will correspond to the selection of certain partial solutions of smaller routes and the subsequent concatenation of these partial solutions to obtain a final solution of the vehicles' routes. The length of the path will upper bound the total latency of the final solution.

We say that $w_\ell$ is an {\em extreme point} of the sequence $(w_1,\ldots,w_n)$ if $(\ell,w_\ell)$ is an extreme-point of the convex hull of $\{(j,w_j): j=1,\ldots,n\}$. Given a point-set $C\sse\R_+^2$, define its {\em lower-envelope curve} 
$f:[\min_{(x,y)\in C}x,\max_{(x,y)\in C}x]\mapsto\R_+$  
by $f(x)=\min\{y: (x,y) \in\conv(C)\}$ where $\conv(C)$ denotes the convex hull of $C$.
We say that $(\ell,w_\ell)$ is a {\em corner point} of the lower-envelope curve of $\{(j,w_j): j=1,\ldots,n\}$ if $w_\ell$ is an extreme point of $(w_1,\ldots,w_n)$.

We have the following results on the concatenation graph:
\begin{theorem}[\cite{GK96,AB10}] \label{thm:concat_graph}
The shortest $1\leadsto n$ path in $\cg(w_1,\ldots,w_n)$ has length at most
$\frac{\mu^*}{2} \sum_{\ell=1}^n w_\ell$ where $\mu^* < 3.5912$ is the solution to $\mu \ln \mu = \mu + 1$. 
Moreover, the shortest path only visits nodes
corresponding to extreme points of $(w_1,\ldots,w_n)$. 
\end{theorem}

\begin{corollary}[Corollary 2.2 in \cite{PS15}] \label{cor:concat_graph}
The shortest $1\leadsto n$ path in $\cg(w_1,\ldots,w_n)$ has length at most
$\frac{\mu^*}{2} \int_{1}^n f(x)dx$, where $f:[1,\ldots,n]\mapsto\R_+$ is the
lower-envelope curve of $\{(j,w_j): j=1,\ldots,n\}$, and only visits nodes
corresponding to extreme points of $(w_1,\ldots,w_n)$. 
\end{corollary}

\section{Point-to-Point Requests} \label{sec:p2p-requests}

In this section, we present polynomial-time constant-factor approximation algorithms for MLP and single-depot/multi-depot $k$-MLP with point-to-point requests. The main idea is to first reduce the instance to a point-request instance in a modified metric and, then, solve linear program~\eqref{eq:lp} and round the fractional optimal solution to nearly optimal routes via the arborescence-packing result and concatenation process described in Section~\ref{sec:lp-framework}. For single-depot $k$-MLP and MLP, we have a lossless reduction to a point-request instance in a {\em directed} metric; for multi-depot $k$-MLP, we reduce to the point-request version in an undirected metric incurring a factor-3 loss. Intuitively speaking, our algorithms find a sequence of frontiers of increasing sizes around the depot(s) and route the vehicle(s) to satisfy the requests in the order determined by the frontiers. See Figure~\ref{fig:algsketch} for a visualization.
\begin{figure}[t]
\centering
\includegraphics[scale=0.25]{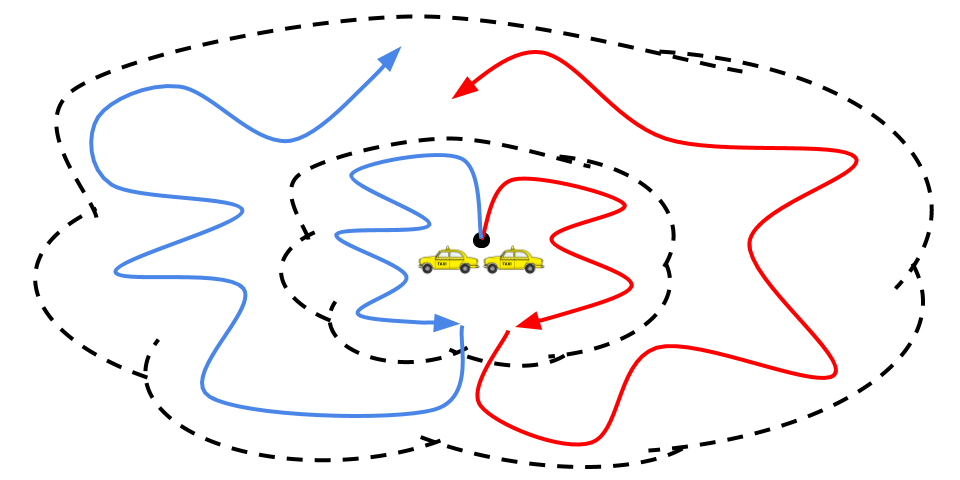}
\caption{A visualization of the structure of routes computed by our algorithms (e.g., Algorithm~\ref{alg:p2p-rounding}). The frontiers are dashed and routes are solid.}
\label{fig:algsketch}
\end{figure}

Our main results are as follows. We prove the single-depot results in Section~\ref{sec:single-depot-kmlp} and the multi-depot ones in Section~\ref{sec:multi-depot-kmlp}.

\begin{theorem} \label{thm:mlp-p2p}
For MLP with point-to-point requests, we can compute a $(\mu^* + \epsilon)$-approximate solution in time $\poly(\inpsize, \frac{1}{\epsilon})$ for any $\epsilon > 0$. %($1.5 \mu^* \approx 5.387$.) 
\end{theorem}

\begin{theorem} \label{thm:kmlp-p2p}
For single-depot $k$-MLP with point-to-point requests, we can compute a $(2.5\mu^* + \epsilon)$-approximate solution in time $\poly(\inpsize, \frac{1}{\epsilon})$ for any $\epsilon > 0$. ($2.5 \mu^* \approx 8.978$.) 
\end{theorem}

\begin{theorem} \label{thm:mult-kmlp-p2p}
For multi-depot $k$-MLP with point-to-point requests, we can compute a $(25.488 + \epsilon)$-approximate solution in time $\poly(\inpsize, \frac{1}{\epsilon})$ for any $\epsilon > 0$.
\end{theorem}

\subsection{Single Depot} \label{sec:single-depot-kmlp}

%We present $2.5\mu^*$-approximation and $1.5\mu^*$-approximation algorithms for the single-depot $k$-MLP and MLP with point-to-point requests, respectively. 
The algorithm leading to Theorem~\ref{thm:kmlp-p2p} is given in Algorithm~\ref{alg:p2p-rounding}; the improved ratio for MLP (Theorem~\ref{thm:mlp-p2p}) is due to a simple observation and the underlying algorithm and analysis are essentially identical. As noted earlier, we transform the given undirected problem instance with point-to-point requests to a point-request instance in a {\em directed} metric (see Step 1 of Algorithm 1) %~\ref{alg:p2p-rounding}) 
and apply the linear programming approach in Section~\ref{sec:lp-framework}. 
%In Step 1 of Algorithm~\ref{alg:p2p-rounding}, the construction of the directed instance $G' = (V', A)$ is given. $G'$ is an instance of single-depot $k$-MLP with point requests in a directed metric. 
The transformation to a directed metric is lossless since there is a one-to-one correspondence between feasible solutions in $G$ and $G'$ and corresponding solutions have the same total latency. 
%Note that the arc lengths $c'$ instead form a directed metric: $c'_{vv} = 0$ for all $v \in V'$ and $c'$ satisfies the triangle inequality. 
Algorithm~\ref{alg:p2p-rounding} works with the directed instance $G'$ in Steps 1--7 and with the given undirected instance $G$ in Steps 7 and 8. 

\begin{figure*}
\begin{minipage}{\textwidth}
{\small \hrule 
\begin{algorithm} \label{alg:p2p-rounding} 
The input is a single-depot $k$-MLP instance $G = (V, E)$ with point-to-point requests and root $r_0$. 
\end{algorithm}
%\vspace{-10pt}
\begin{labellist}[]
\item Define a directed graph $G' = (V', A)$ on $n+1$ nodes corresponding to the root/requests and arcs $(v_i,v_j)$ with length $c'_{v_i v_j} = c_{d_i s_j} + c_{s_j d_j}$ for all requests $i$ and $j$. Treat root $r_0$ as the $0$-th request $(s_0, d_0)$ with $r_0 = s_0 = d_0$. Let $r$ denote the root in $G'$.  
Let $\Time = 4n \lb$ where $\lb := \max_{v \in V} c_{r_0 v}$. By scaling and rounding the $c$-metric (e.g., \cite{AK03}), we may assume that $\Time=\poly\bigl(\inpsize,\frac{1}{\epsilon}\bigr)$ losing a $(1+\epsilon)$-factor.

\item %Compute a $(1+\epsilon)$-approximate 
Compute an optimal solution $(x, z)$ to \eqref{eq:lp} for the instance given by $G'$ and $\{c'_a\}$. Let $x'_{v,t}=\sum_{i}x^i_{v,t},\ z'_{a,t}=\sum_{i}z^i_{a,t'}$ for all $v, a, t$.

\item Initialize $C\assign\{(1,0)\}$, $\Qc\assign\es$. Let $K$ be such $Kz'_{a,t}$ is an
integer for all $a, t$. 
For $t\in[\Time]$, define $S(t)=\{u\in V: \sum_{t'=0}^t x'_{u,t'}>0\}$. 
(Note that $r\in S(t)$ for all $t>0$.)

\item For all $t=1,\ldots,\Time$, do the following. 
Apply Theorem~\ref{thm:arbpoly} on the digraph $D$ with edge weights $\{Kz'_{a,t}\}_{a\in A}$ and
integer $K$ (and root $r$) to obtain a weighted arborescence family
$(\gm_1,Q^t_1),\ldots,(\gm_q,Q^t_q)$. 
For each arboresence $Q^t_\ell$ in the family, %which we view as a tree, 
add the point $\bigl(|V(Q^t_\ell)\cap S(t)|,\frac{3c'(Q^t_\ell)}{k}+2t\bigr)$ to $C$, and add $Q^t_\ell$ to $\Qc$. 

\item For all $\ell=1,\ldots,n$, compute $f(\ell)$ where 
$f:[1,n]\mapsto\R_+$ is the lower-envelope curve of $C$.
We show in Lemma~\ref{lem:lpsbound} that for every corner point $\bigl(\ell,f(\ell)\bigr)$ of
$f$, 
there is some tree $Q^*_\ell\in\Qc$ and some time $t^*_\ell$
such that $\ell=|V(Q^*_\ell)\cap S(t^*_\ell)|$, 
$f(\ell)=\frac{3c'(Q^*_\ell)}{k}+2t^*_\ell$, and 
$\max_{v\in Q^*_\ell\cap S(t^*_\ell)}c'_{r v}\leq t^*_\ell$. 

\item Find a shortest $1\leadsto n$ path $P_{C}$ in the concatenation graph $\cg(f(1),\ldots,f(n))$.

\item For every node $\ell>1$ on $P_{C}$, do the following. 
%Double and take a Eulerian tour from $r$ on $Q^*_\ell$ using only its edges while ignoring the directions.
Do a DFS-traversal from $r$ on $Q^*_\ell$ to obtain a cycle $Z_\ell$, in the undirected sense, that uses each arc of $Q^*_\ell$ twice.
Break $Z_\ell$ into $k$ segments, each of $c'$-length at most $2c'(Q^*_\ell)/k$. 
%while preserving the direction.
For each segment, create the corresponding tour $Z_{i,\ell}$ in $G$, starting from and ending at $r_0$, that satisfies only the requests that have their (first) appearance within the segment in the Eulerian tour and in the order of first appearances. Skip requests that are not in $S(t^*_\ell)$. With probability $\frac{3}{4}$, we traverse $Z_{i, \ell}$ in the correct direction and with probability $\frac{1}{4}$, in the opposite direction (satisfying the same set of requests but in the reverse order).
This yields a collection of $k$ tours $Z_{1,\ell}, \ldots, Z_{k,\ell}$.

\item For every $i=1,\ldots,k$, concatenate the tours $Z_{i,\ell}$ for nodes $\ell$ on $P_C$ to obtain vehicle $i$'s route, shortcutting if a request has already been satisfied. 
\end{labellist}
\hrule
}
\end{minipage}
\end{figure*}

The analysis closely resembles the one in \cite{PS15} but is more involved due to ``directionality'' of the requests since a vehicle serves a request by going from the source to the destination, not the other way around.
Lemma~\ref{lem:lpsbound} relates the lower-envelope curve $f$ of $C$ in Step 5 and the objective value of the fractional optimal solution in Step 2 as in \cite{PS15}, since this only requires nonnegative edge costs.

\begin{lemma} \label{lem:lpsbound}
(i) $\int_1^nf(x)dx\leq 5\sum_{u\in V,t\in[\Time]}tx'_{u,t}$. 
(ii) If $\bigl(\ell,f(\ell)\bigr)$ is a corner point of $f$, then there is a tree $Q^*_\ell$
and time $t^*_\ell$ satisfying the properties stated in Step 5 in Algorithm~\ref{alg:p2p-rounding}.
\end{lemma}
\begin{proof}
This is a modified version of Lemma 6.3 in \cite{PS15} and has essentially the same proof. The only difference is that we add the point $\bigl(|V(Q^t_\ell)\cap S(t)|,\frac{3c'(Q^t_\ell)}{k}+2t\bigr)$ to $C$. By Theorem~\ref{thm:arbpoly}, $\E[c'(Q_\ell^t)] \leq kt$ and, hence, $\E\left[\frac{3c'(Q^t_\ell)}{k}+2t \right] \leq 5t$. The rest follows the same proof and the lemma follows. 
\end{proof}

%Furthermore, we have the following:

\begin{lemma} \label{lem:break}
If $(\ell, f(\ell))$ is a corner point of $f$, each of the tours $Z_{1, \ell}, \ldots, Z_{k, \ell}$ created in Step 7 has length at most $\frac{2 c'(Q^*_\ell)}{k} + 2 t^*_\ell$. In particular, the part of each tour without the first and last edges connecting to the root has length at most $\frac{2 c'(Q^*_\ell)}{k}$.
\end{lemma}
\begin{proof}
When breaking the cycle $Z_\ell$, a break point can be either on a node or on an edge. In the former case, the node appears in two consecutive segments, and in the latter case, the edge is removed. In either cases, all the nodes of $Q^*_\ell$ will be contained in at least one segment and the corresponding request will be covered in one of the $k$ tours $Z_{1,\ell}, \ldots, Z_{k,\ell}$ by the construction below.

Consider a segment and the corresponding tour $Z_{i,\ell}$. In the cycle $Z_\ell$, we visit a node through a ``correctly'' oriented edge (i.e., in the same direction as the tour) for the first visit and through ``incorrectly'' oriented edges (i.e., in the opposite direction) for subsequent visits, and consequently, the segment consists of both correctly and incorrectly oriented edges. Without loss of generality, let $R_1, \ldots, R_q$ be the requests corresponding to the vertices that are visited for the first time within the given segment in $Z_\ell$. Then, $Z_{i,\ell} = r_0 s_1 d_1 \cdots s_q d_q r_0$. For $P = d_1 s_2 d_2 \cdots s_q d_q$, we show $c(P) \leq \frac{2c'(Q^*_\ell)}{k}$. By the definition of $c'$, each arc in $G'$ corresponds to a path of length $2$ in the original graph $G$. We replace the edges of the segment with corresponding paths of length $2$ of the form $d_i s_j d_j$. The resulting path in $G$ has length at most $\frac{2 c'(Q^*_\ell)}{k}$ and contains the nodes of $P$ as a subsequence, not necessarily contiguously. We shortcut to get $P$ exactly where shortcutting involves going directly from $d_i$ to $s_{i+1}$ for some $i$ and truncating the beginning or end. Since $G$ has metric edge costs $\{c_e\}$, it follows that $c(P) \leq \frac{2 c'(Q^*_\ell)}{k}$. 

For the two paths $r_0 s_1 d_1$ and $d_q r_0$ that complete $P$ into $Z_{i,\ell}$, we upper bound the length of each path by $t^*_\ell$. Let $r$ be the root in $G'$ and $v_1, v_q \in V'$ be the vertices corresponding to $d_1$ and $d_q$, respectively. Since $v_1$ and $v_q$ are in $Q^*_\ell \cap S(t^*_\ell)$, $c'_{r v_1} \leq t^*_\ell$ and $c'_{r v_q} \leq t^*_\ell$. Then, $c_{r_0 s_1} + c_{s_1 d_1} = c'_{r v_1} \leq t^*_\ell$ and $c_{r_0 d_q} \leq c_{r_0 s_q} + c_{s_q d_q} = c'_{r v_q} \leq t^*_\ell$. It follows that $c(Z_{i,\ell}) = c(P) + c_{r_0 s_1} + c_{s_1 d_1} + c_{r_0 d_q} \leq \frac{2 c'(Q^*_\ell)}{k} + 2 t^*_\ell$. 
\end{proof}

%Finally, we can prove Theorems~\ref{thm:kmlp-p2p} and \ref{thm:mlp-p2p}.

\begin{proof}[Proof of Theorem~\ref{thm:kmlp-p2p}]
We claim that the solution returned by Algorithm~\ref{alg:p2p-rounding} has total latency at most the length of $P_C$ in the concatenation graph $\cg(f(1), \ldots, f(n))$. By Corollary~\ref{cor:concat_graph} and Lemma~\ref{lem:lpsbound} part (i), the length of $P_C$ is at most $\frac{\mu^*}{2}\int_{1}^n f(x)dx \leq \frac{5 \mu^*}{2} \sum_{v, t} t x'_{v,t} = 2.5 \mu^* \sum_{v, i, t} t x^i_{v,t}$. Also, $\sum_{v,i,t}t x^i_{v,t}$ is at most $(1+\epsilon)$ times the optimal latency, where the $(1+\epsilon)$-factor is due to scaling and rounding.
%Since the solution has total latency at most the length of $P_C$ and $\sum_{v, i, t} t x^i_{v, t}$ is the objective value of a $(1+\epsilon)$-approximate solution to \eqref{eq:lp}, the theorem would follow.

We now prove the claim. More specifically, we show that the total latency is at most the length of $P_C$ in $\cg(f(1), \ldots, f(n))$ by induction. By Theorem~\ref{thm:concat_graph} and Lemma~\ref{lem:lpsbound}, there exist $Q^*_\ell$ and $t^*_\ell$ satisfying the properties stated in Step 5 for each corner point $(\ell, f(\ell))$ on $P_C$. By Lemma~\ref{lem:break}, the tours that would be created from $Q^*_\ell$ have length at most $\frac{2 c'(Q^*_\ell)}{k} + 2 t^*_\ell$ in the correct direction and cover all the requests corresponding to the vertices of $Q^*_\ell$. In particular, the part of each tour without the first and last edges connecting to the root has length at most $\frac{2 c'(Q^*_\ell)}{k}$ in the correct direction. 

Suppose inductively that we have covered at least $o$ requests by the partial solution constructed by concatenating tours corresponding to the nodes on $P_C$ up to and including $o$. Assume we next take an edge $(o, \ell)$ and consider the additional contribution to the total latency when we concatenate $Z_{i,\ell}$ to vehicle $i$'s route for $i \in [k]$. Note the resulting partial solution covers at least $\ell$ requests.

Each request covered in the current concatenation step incurs additional latency of at most $\frac{f(\ell)}{2}$ in expectation. To see this, let $L$ be the length of the part of a tour $Z_{i, \ell}$ without the root $r_0$ and $L'$ be the length up to the destination of a request covered by the tour. With probability $\frac{3}{4}$, the additional latency incurred for the request is $t_\ell + L'$. With probability $\frac{1}{4}$, the additional latency is at most $t_\ell + 3(L - L')$ because the length of a traversal in the opposite direction is upper bounded by 3 times the length of the corresponding traversal in the correct direction (accounting the $s_i$-$d_i$ portions 3 times). For example, if $s_1 d_1 \ldots s_n d_n$ is a traversal in the correct direction, the reverse traversal is a shortcut version of $d_n s_n d_n s_n \ldots d_1 s_1 d_1 s_1$ which has length at most 3 times that of $s_1 d_1 \ldots s_n d_n$. In expectation, the additional latency incurred for the request is at most $\frac{3}{4}(t_\ell + L') + \frac{1}{4} (t_\ell + 3L - 3L') = t_\ell + \frac{3L}{4} \leq \frac{f(\ell)}{2}$ as $L \leq \frac{2 c'(Q^*_\ell)}{k}$.

By a similar argument, each request that is still uncovered after concatenation incurs additional latency of at most $f(\ell)$ in expectation. The traversal in the correct direction has length at most $2 t_\ell + L$ and the one in the opposite direction has length at most $2 t_\ell + 3L$. In expectation, the additional latency is at most $2 t_\ell + \frac{3L}{2} \leq f(\ell)$.

There are exactly $\ell$ requests covered by the tours $Z_{1, \ell}, \ldots, Z_{k, \ell}$ and at most $\ell - o$ new requests are covered in the current concatenation step. Each of these incurs additional latency at most $\frac{f(\ell)}{2}$. At most $n - \ell$ requests remain to be covered and they each incur additional latency at most $f(\ell)$. The overall increase in the total latency is at most $\frac{f(\ell)}{2} (\ell - o) + f(\ell) (n - \ell) = f(\ell) \left( n - \frac{o + \ell}{2}\right)$ which is equal to the length of the edge $(o, \ell)$ in $\cg(f(1), \ldots, f(n))$. By induction, the total latency is at most the length of $P_C$. 

For the running time, 
%we let $\Time = 4n \lb$ where $\lb := \max_{v \in V} c_{r_0 v}$ where $V$ are the nodes in the original undirected instance. Clearly, $\Time$ is an upper bound on the maximum latency of a request, and we assume it is polynomially bounded at the expense of a $(1+\epsilon)$-factor loss by standard scaling and rounding (e.g., \cite{AK03}). 
since $\Time=\poly(\inpsize,\frac{1}{\epsilon})$, we can design a separation oracle for \eqref{jcov3}, which is a min-cut algorithm, and solve \eqref{eq:lp} in time $\poly(\inpsize, \frac{1}{\epsilon})$. 

%TODO: what is the scaling-and-rounding technique? check $\Time$ and $\log \Time$ are polynomially bounded?
% Note needed for this application. For some extensions, we have only that $\log \Time$ is polynomially bounded. In this case, we can restrict to exponentially separated time points $\TS:=\{\Time_0,\ldots,\Time_D\}$ where $\Time_j=\ceil{(1+\e)^j}$, and $D=O(\log\Time)=\poly(\inpsize)$ is the smallest integer such that $\Time_D\geq\Time$. Then, we solve as described in Section 7. In our problem, we can set $\Time = \sum_{i} c_{s_i t_i} + 2n \lb$ as the upper bound and let $\log \Time$ be polynomially bounded.
\end{proof}

\begin{proof}[Proof of Theorem~\ref{thm:mlp-p2p}]
We follow the same analysis above for Algorithm~\ref{alg:p2p-rounding}. The improvement of the approximation ratio comes from the fact that we do not need to break the cycle $Z_\ell$ into $k$ tours in Step 7 since we have exactly one vehicle. 
Moreover, we note that the length of a traversal of $Z_\ell$ in {\em either} direction is at most $2c'(Q^*_\ell)$. In Step 4, we now add the point $\bigl(|V(Q^t_\ell)\cap S(t)|, 2c'(Q^t_\ell)\bigr)$ to $C$ to leverage this fact. 
%In Lemma~\ref{lem:lpsbound} part (i), this yields the improved bound $\int_1^nf(x)dx\leq 3\sum_{u\in V,t\in[\Time]}tx'_{u,t}$. By Corollary~\ref{cor:concat_graph}, the length of $P_C$ is at most $\frac{\mu^*}{2}\int_{1}^n f(x)dx \leq 1.5 \mu^* \sum_{v, i, t} t x^i_{v,t}$. 
Following the proof of Theorem~\ref{thm:kmlp-p2p}, we now therefore obtain a $\mu^*$-approximation.
\end{proof}

We remark that for single-depot $k$-MLP, we can avoid the additive $\epsilon$ in the approximation factors by utilizing the more combinatorial approach presented in \cite{PS15}. In this approach, we obtain the arborescences directly without solving an LP. Specifically, for each $i=1,\ldots,n$, we aim to find the least $c'$-cost arborescence spanning at least $i$ requests (which in turn utilizes arborescence packings). We then again use the concatenation graph to select a subset that will be converted into tours which are then concatenated. This is closer to the approach we follow in our experimental results.

\subsection{Multiple Depots} \label{sec:multi-depot-kmlp}

For the more general multi-depot $k$-MLP with point-to-point requests, we provide a %constant-factor reduction of the problem to the multi-depot $k$-MLP with point requests such 
reduction showing that an $\alpha$-approximation algorithm for the point-request version on undirected metrics yields a $3 \alpha$-approximation algorithm for the point-to-point request version on undirected metrics. 
%The resulting point-request instance has undirected edges that form a metric space and, hence, any algorithm for the point-request version can be used to solve it.
More specifically, the following result immediately yields Theorem~\ref{thm:mult-kmlp-p2p} using the $(8.496+\epsilon)$-approximation for multi-depot $k$-MLP with point-requests in~\cite{PS15}. 
\if\conference1
The proofs in the remainder of this section are deferred to the full version of the paper.
\fi
%We now describe the constant-factor reduction. 
%and then prove Theorem~\ref{thm:mult-kmlp-p2p}.

\begin{lemma} \label{lem:backtrack-reduction}
Let $\OPT$ and $\OPT'$ denote respectively the optimal values of the multi-depot $k$-MLP instance with point-to-point requests, and the multi-depot $k$-MLP instance with point requests obtained by the reduction under the $c'$ edge costs (described below). Then $\OPT\leq\OPT'\leq 3\cdot\OPT$. Hence, an $\alpha$-approximate solution to the point-requests instance yields a $3\alpha$-approximate solution to the point-to-point requests instance.
\end{lemma}

The main idea of the factor-$3$ reduction is to 
%transform graph $G$ into another graph $G' = (V', E')$ with a distance function $c'$ by 
identify the requests by their source nodes and incorporate the request-specific distances $c_{s_i d_i}$ into distances between sources in a symmetric way. We treat the root nodes as dummy requests with the identical source and destination nodes for the reduction. More specifically, $G'$ contains only root nodes and nodes representing the requests, and the distance between $v_i$ and $v_j$ in $G'$ is
$c'_{v_i v_j} = c_{s_i s_j} + c_{s_i d_i} + c_{s_j d_j}$.
%%\[
%%c'_{v_i v_j} = c_{s_i s_j} + c_{s_i d_i} + c_{s_j d_j} \,.
%%\]
It is easy to see that $(G',c')$ is a metric.
%The resulting graph $G'$ is a metric space. 
We then solve the resulting problem instance using an algorithm for the point requests and convert the computed paths into ones in the original problem instance. Via this reduction, observe that the latency of a rooted path in the $(G',c')$ instance corresponds to the latency of a particular type of path, that we call a {\em backtracking path}, in the original problem instance, where a vehicle satisfies each request $i$ by going from $s_i$ to $d_i$ and then returning to $s_i$ before moving onto the next request. The backtracking paths are for the purpose of analysis only and we shortcut these paths to obtain the final routes for the vehicles. 
%
%We define backtracking paths formally and prove its properties. For ease of exposition, we focus on a single path from, say, root $r_0$ but the results also apply when there are multiple paths starting from multiple depots. We define a {\em backtracking path} to be the kind that goes through the requests by satisfying each request $R_i$ by going from $s_i$ to $d_i$ and then coming back to $s_i$ before moving onto the next request. 
Clearly, there are bijective mappings between request-orderings and paths and backtracking paths. If the order of the requests is $R_1 \cdots R_n$, which equivalently corresponds to a path $P$, the corresponding backtracking path is 
\[
%P=r_0 s_1 d_1 s_2 d_2 \cdots s_n d_n\,, \quad
P_b = r_0 s_1 d_1 s_1 s_2 d_2 s_2 \cdots s_n d_n \,.
\]
%We use $P$ and $P_b$ to denote the corresponding path and backtracking path with respect to the same order of requests.
If the order of the requests is $R_1 \cdots R_n$ on the route $P$ of a vehicle starting from, say, root $r_0$, the latency of the $i$-th request on the corresponding backtracking path $P_b$ is determined as
\[
\LAT(P_b, i) = \begin{cases} \LAT(P_b, i-1) + c_{d_{i-1} s_{i-1}} + c_{s_{i-1} s_i} + c_{s_i d_i}, & i > 1 \\
0, &i = 0 \end{cases} \,.
\]
The total latency of $P_b$ is $\LAT(P_b) = \sum_{i=1}^n \LAT(P_b, i)$. In particular, the following property is useful in proving Lemma~\ref{lem:backtrack-reduction}.

\begin{lemma} \label{lem:backtrack}
For any path $P$ and corresponding backtracking path $P_b$, we have
$\LAT(P) \leq \LAT(P_b) \leq 3 \LAT(P)$.  Further, the factor of $3$ is tight.
\end{lemma}
\if\conference0
\begin{proof}
Without loss of generality, let $R_1 R_2 \cdots R_n$ be the ordering of requests on $P$. Let $\LAT(P,i)$ denote the latency of the $R_i$ on $P$. Define $\LAT(P,i)=0$ for all $i\leq 0$.
%and consider the corresponding path $P$ and backtracking path $P_b$ from root $r_0$.
%\noindent {\em (Lower bound)}
We show by induction that $\LAT(P, i) \leq \LAT(P_b, i)\leq 2\LAT(P,i-1)+\LAT(P,i)$ for all $i$. Since $2\LAT(P,i-1)+\LAT(P,i)\leq 3\LAT(P,i)$, summing over all $i$ implies the statement. 
%Then, it would follow that
%\[
%\LAT(P) = \sum_{i=1}^n \LAT(P, i) \leq \sum_{i=1}^n \LAT(P_b, i) = \LAT(P_b) \,.
%\]

The base case when $i=1$ holds since $\LAT(P, 1) = \LAT(P_b, 1) = c_{r_0 s_1} + c_{s_1 d_1}$. Suppose inductively the claim holds for $i-1$. 
%We now assume $\LAT(P, i-1) \leq \LAT(P_b, i-1)$ and show $\LAT(P, i) \leq \LAT(P_b, i)$. 
Then
\begin{align*}
\LAT(P, i) &= \LAT(P, i-1) + c_{d_{i-1} s_i} + c_{s_i d_i} \\
	&\leq \LAT(P_b, i-1) + c_{d_{i-1} s_{i-1}} + c_{s_{i-1} s_i} + c_{s_i d_i} \\
	&= \LAT(P_b, i) \,,
\end{align*}
where the inequality follows from the induction hypothesis and the triangle inequality. %Hence, the induction holds and the lower bound is proved.
%
%\noindent {\em (Upper bound)}
%Similarly, we show that $\LAT(P_b, i) \leq 2 \LAT(P, i-1) + \LAT(P, i)$ for all $i$ by induction. Then,
%\begin{multline*}
%\LAT(P_b) = \sum_{i=1}^n \LAT(P_b, i) \leq \sum_{i=1}^n \left( 2 \LAT(P, i-1) + %\LAT(P, i) \right) \\
%\leq 3 \sum_{i=1}^n \LAT(P, i) = 3 \LAT(P) \,.
%\end{multline*}
%The base case holds since $\LAT(P_b, 1) = \LAT(P, 1)$ and $\LAT(P, 0) = 0$. Suppose $\LAT(P_b, i-1) \leq 2 \LAT(P, i-2) + \LAT(P, i-1)$. Then,
For the upper bound, we have that
\begin{equation*}
\begin{split}
\LAT(P_b, i) &= \LAT(P_b, i-1) + c_{d_{i-1} s_{i-1}} + c_{s_{i-1} s_i} + c_{s_i d_i} \\
	&\leq 2 \LAT(P, i-2) + \LAT(P, i-1) + (c_{s_{i-1} d_{i-1}} + c_{d_{i-1} s_i}) + \\
	&\quad\quad c_{d_{i-1} s_{i-1}} + c_{s_i d_i} \\
	&= 2\left( \LAT(P, i-2) + c_{s_{i-1} d_{i-1}} \right) + \LAT(P, i-1) + \\
	&\quad\quad c_{d_{i-1} s_i} + c_{s_i d_i} \\
	&\leq 2 \LAT(P, i-1) + \LAT(P, i) \,.
\end{split}   
\end{equation*}
The first inequality follows from the induction hypothesis and the triangle inequality; and the second inequality follows from the recurrence definitions of $\LAT(P, i-1)$ and $\LAT(P, i)$. 
\end{proof}
\fi

\if\conference0
To see that the factor of 3 is tight, consider the following instance with $n$ requests $R_1, \ldots, R_n$ which are to be satisfied in that order. The sources are evenly spaced on a line at the distance interval of 1 and the destinations coincide with the sources such that $r_0 = s_1$, $d_1 = s_2$, $d_2 = s_3$, etc. For the path $P$, $\LAT(P, i) = i$ for all $i$ and the total latency is $\LAT(P) = \frac{1}{2} n^2 + \frac{1}{2} n$. For the backtracking path $P_b$, $\LAT(P_b, i) = 3(i-1) + 1$ for all $i$ and the total latency is $\LAT(P_b) = \sum_{i=1}^n \LAT(P_b, i) = \frac{3}{2} n^2 - \frac{1}{2} n$. Since $\lim_{n\rightarrow \infty} \frac{\LAT(P_b)}{\LAT(P)} = 3$, the ratio can be arbitrarily close to 3. 
\fi

\if\conference0
We can now prove the approximation guarantee via the factor-$3$ reduction.
\begin{proof}[Proof of Lemma~\ref{lem:backtrack-reduction}]
Let $P^1,\ldots,P^k$ be an optimal solution to the original multi-depot $k$-MLP instance with point-to-point requests. 
%Note that the total latency of a solution for the modified instance (with the $c'$-edge costs) is precisely the total latency of the corresponding backtracking paths. 
Let $Q^1_b,\ldots,Q^k_b$ be backtracking paths corresponding to an optimal solution to the resulting multi-depot $k$-MLP instance with point requests (and the $c'$-edge costs). 

By Lemma~\ref{lem:backtrack}, the paths $Q^1_b,\ldots,Q^k_b$ map to paths for the original instance of no greater total latency, so $\OPT\leq\OPT'$. Again, by Lemma~\ref{lem:backtrack}, the paths $P^1,\ldots,P^k$ map to backtracking paths of total latency at most $3\cdot\OPT$, so $\OPT'\leq 3\cdot\OPT$.

The second statement follows immediately from the first one since any solution to the modified instance with point requests yields a solution to the original instance of no greater total latency. 
\end{proof}
\fi

\section{Release-Time Constraints} \label{sec:kmlp-release-results}

We show constant-factor approximation algorithms for the minimum latency problems with release-time constraints for both point requests and point-to-point requests. We incorporate the release-time constraints into the linear program \eqref{eq:lp} and solve the resulting linear program optimally as before. When rounding an optimal fractional solution to an integral solution, we follow the same analysis steps in Section~\ref{sec:p2p-requests} but need to account for the release times. Our analysis can be easily modified to satisfy the release-time constraints with little or no extra cost in terms of approximation guarantees. 

For point requests, we have the following results with similar approximation ratios as for the variants without release times. We prove them in Sections~\ref{sec:kmlp-release} and \ref{sec:kmlp-p2p-release}, respectively:

\begin{theorem} \label{thm:kmlp-release}
For point requests with release times, we obtain the following %constant-factor 
approximation guarantees in time $\poly(\inpsize, \frac{1}{\epsilon})$ for any $\epsilon > 0$:
\begin{enumerate}[(1)]
%\item A $(1.5\mu^* + \epsilon)$-approximation for MLP; %solution. 
%($1.5 \mu^* \approx 5.387$) 
\item A $(2\mu^* + \epsilon)$-approximation for MLP and single-depot $k$-MLP; ($2 \mu^* \approx 7.183$)
\item A $(13.728 + \epsilon)$-approximation for multi-depot $k$-MLP.
\end{enumerate}
\end{theorem}

%\begin{theorem} \label{thm:kmlp-release}
%For single-depot $k$-MLP with point requests with release times, we can compute a $(2\mu^* + \epsilon)$-approximate solution in time $\poly(\inpsize, \frac{1}{\epsilon})$ for any $\epsilon > 0$. ($2 \mu^* \approx 7.183$.) 
%\end{theorem}
%
%\begin{theorem} \label{thm:kmlp-release}
%For MLP with point requests with release times, we can compute a $(1.5\mu^* + \epsilon)$-approximate solution in time $\poly(\inpsize, \frac{1}{\epsilon})$ for any $\epsilon > 0$. ($1.5 \mu^* \approx 5.387$.) 
%\end{theorem}
%
%\begin{theorem} \label{thm:mult-kmlp-release}
%For multi-depot $k$-MLP with point requests with release times, we can compute a $(16.994 + \epsilon)$-approximate solution in time $\poly(\inpsize, \frac{1}{\epsilon})$ for any $\epsilon > 0$. 
%\end{theorem}

For point-to-point requests, we also have constant-factor approximation algorithms:

\begin{theorem} \label{thm:kmlp-p2p-release}
For point-to-point requests with release times, we have the following %constant-factor 
approximation guarantees in time $\poly(\inpsize, \frac{1}{\epsilon})$ for any $\epsilon > 0$:
\begin{enumerate}[(1)]
\item A $(2\mu^* + \epsilon)$-approximation for MLP; ($2\mu^* \approx 7.183$) 
\item A $(2.5\mu^* + \epsilon)$-approximation for single-depot $k$-MLP; ($2.5 \mu^* \approx 8.978$) 
\item A $(41.184 + \epsilon)$-approximation for multi-depot $k$-MLP.
\end{enumerate}
\end{theorem}

%\begin{theorem} \label{thm:kmlp-p2p-release}
%For single-depot $k$-MLP with point requests with release times, we can compute a $(2.5\mu^* + \epsilon)$-approximate solution in time $\poly(\inpsize, \frac{1}{\epsilon})$ for any $\epsilon > 0$. ($2.5 \mu^* \approx 8.978$.) 
%\end{theorem}
%
%\begin{theorem} \label{thm:mlp-p2p-release}
%For MLP with point-to-point requests, we can compute a $(2\mu^* + \epsilon)$-approximate solution in time $\poly(\inpsize, \frac{1}{\epsilon})$ for any $\epsilon > 0$. ($2 \mu^* \approx 7.183$.) 
%\end{theorem}
%
%\begin{theorem} \label{thm:mult-kmlp-p2p-release}
%For multi-depot $k$-MLP with point requests with release times, we can compute a $(16.994 + \epsilon)$-approximate solution in time $\poly(\inpsize, \frac{1}{\epsilon})$ for any $\epsilon > 0$. 
%\end{theorem}

As in the case without release times, the total latency objective is computed with respect to the beginning of the paths at the root nodes. This is analogous to the total-completion-time objective that is widely used in the scheduling literature (see, e.g.,~\cite{S06}). While, admittedly, measuring the latency of a request $R_i$ by (time taken to reach $d_i$) $-T_i$, which accounts for the actual time $R_i$ spends in the system, yields a more natural objective that is akin to the {\em flow-time} objective in scheduling, the resulting optimization problem is much harder to approximate (as is the case in scheduling). Therefore, following much of the work in scheduling, we consider our latency objective, which is a reasonable first step in studying ride-sharing and delivery problems with release times.

To distinguish from the latency objective without release times, we use $\LAT^+$ to denote latencies and incorporate the release times by making vehicles wait at the source of a request if it has not been released. For example, for path $P = r_0 s_1 d_1 \cdots s_n d_n$, note that the latency of the $i$-th request on $P$ is 
\[
\LAT^+(P, i) = \begin{cases} \max\{ \LAT^+(P, i-1) + c_{d_{i-1} s_i}, T_i \} + c_{s_i d_i}, & i > 1 \\
	0, & i = 0 \end{cases} \,.
\]
Similarly, the length of a path/tour needs to correspond to the traversal time subject to release-time constraints. Given a sequence of requests to be satisfied, the traversal time for the sequence consists of {\em driving times} for moving from one location to another and {\em waiting times} for staying fixed at a location before a request is released. We can still concatenate and shortcut in the sense that the total traversal time of concatenation of two paths meeting at a node is at most the sum of traversal times of the paths, and the traversal time of a shortcut path is at most that of the original path. The algorithms and their analyses stay essentially the same for traversal times as for lengths of paths/tours. 
%we will use traversal times and lengths interchangeably when convenient. 
 
\subsection{Point Requests: Proof of Theorem~\ref{thm:kmlp-release}} \label{sec:kmlp-release}

%In this section, we prove Theorem~\ref{thm:kmlp-release}. The main idea is to 
\paragraph*{Part (1).}
We incorporate the following release-time constraints into the linear program \eqref{eq:lp} and solve for an optimal fractional solution as before:
\begin{equation} \label{releasetime}
x^i_{v,t}=0\ \text{if $T_i>t$} \qquad \forall v,t,i \,.
\end{equation}
(Since we have point requests, the $c'$ edge costs used in Algorithm~\ref{alg:p2p-rounding} are simply the bidirected $c$ costs.)
We now need to incorporate the $R_i$'s in our upper bound $\Time$,
%set $\Time=\max_i R_i+n\max_vc_{rv}$; 
but $\log\Time$ is still polynomially bounded and we can obtain a $(1+\epsilon)$-approximate solution to the LP.
Then, we round the fractional solution as in Algorithm~\ref{alg:p2p-rounding} by creating a set of tours starting and ending at the root 
%(for the single-depot variants) 
and concatenating them. 

% \begin{proof}[Proof of Theorem~\ref{thm:kmlp-release}]
% {\em (Part 1)}
% The proof relies on Corollary 6.4 in \citet{PS15} which is the counterpart of Theorem~\ref{thm:mlp-p2p} and shows the approximation ratio of $\mu^*$ for MLP with point requests without release times. For point requests without release times, we would be adding the point $\bigl(|V(Q^t_\ell)\cap S(t)|, 2c(Q^t_\ell)\bigr)$ to $C$ for the original distances $c$ in Step 4 of the point-request variant, given in \cite{PS15}, of Algorithm~\ref{alg:p2p-rounding}.

% With release times, the actual traversal time can be more than the quantity $2c(Q^t_\ell)$ but by at most $t$ because $S(t)$ contains only requests with release times at most $t$. In other words, the total driving time is at most $2c(Q^t_\ell)$ and the total waiting time at most $t$. We instead add points $\bigl(|V(Q^t_\ell)\cap S(t)|, 2c(Q^t_\ell) + t\bigr)$ and would have the improved bound $\int_1^nf(x)dx\leq 3\sum_{u\in V,t\in[\Time]}tx'_{u,t}$ in Lemma~\ref{lem:lpsbound} part (i). As in Theorem~\ref{thm:mlp-p2p}, we get the approximation guarantee of $1.5 \mu^*$.

% {\em (Part 2)}
In the analysis, we upper bound the traversal times of the tours using the same upper bound (or nearly the same) we have used for the lengths of the tours, and, hence, obtain the same or similar approximation guarantees with respect to the total latency objective as for the variants without release times.
Specifically, Theorem 6.1 in \citet{PS15}, which is the equivalent of Theorem~\ref{thm:kmlp-p2p} for point requests with release times, obtains an approximation ratio of $2\mu^*$ for single-depot $k$-MLP with point requests without release times. %We obtain the same approximation ratio in the case with release times. 
We observe that the same upper bound of $\frac{2c(Q^*_\ell)}{k} + 2t^*_\ell$ on the $c$-length of a tour obtained for time-point $t^*_\ell$ used in Theorem 6.1 in \cite{PS15} also upper bounds the traversal time of any tour $Z_{i,\ell}$ that we obtain in Step 7 of Algorithm~\ref{alg:p2p-rounding}. This is because the upper bound of $t^*_\ell$ on the lengths of the edges connecting the ends of $Z_{i,\ell}$ to the root also upper bounds the waiting time portion of the traversal:
%More specifically, the upper bound , say, tour $Z_{i, \ell}$ assumes the upper bound of $t^*_\ell$ on the distance from the root to the first request on $Z_{i,\ell}$. 
since all requests on $Z_{i, \ell}$ have release time at most $t^*_\ell$ by the release-time constraints \eqref{releasetime}, $t^*_\ell$ also accounts for the total waiting time of the traversal of $Z_{i,\ell}$. Thus, we do not incur any extra cost; the rest of the analysis is the same as that in~\cite{PS15}, and the approximation guarantee follows.

\paragraph*{Part (2).}
For this result, we use Algorithm 1 in \cite{PS15} and its analysis, which uses a related but different linear program. We incorporate the release-time constraints \eqref{releasetime} as before. Due to the space constraints, we omit the linear program and refer to \cite{PS15} for details. As before, we upper bound the traversal times of tours obtained for some time-point $t$ by (the total driving time of at most $t$ to visit the requests in the tour) + $t$, since $t$ is an upper bound on the waiting time incurred as all requests on the tour have release times at most $t$. 
%the traversal times can be at most twice of the corresponding lengths by simply 
Consequently, following the analysis in \cite{PS15}, one can argue that for any constant $1<c<e$, we obtain an approximation ratio of at most $\frac{(2c+1)(1-e^{-1})}{\ln c(1-ce^{-1})}+\epsilon$.%
\footnote{Claim 5.2 (iii) in~\cite{PS15} changes as follows: conditioned on the random offset $h$, the expected latency of a node $v$ first covered in iteration $j$ of the algorithm is now at most $(1+\epsilon)(3t_0+3t_1+\ldots+3t_{j-1}+2t_j)\leq\frac{(1+\epsilon)(2c+1)}{c-1}\cdot t_j$}
%that is twice that obtained in \cite{PS15}, which is equal to $2* 8.497 = 16.994$.
Taking $c=1.58726$, we obtain an approximation ratio of at most $13.7272+\epsilon$.
%\end{proof}

\subsection{Point-to-Point Requests: Proof of Theorem~\ref{thm:kmlp-p2p-release}} \label{sec:kmlp-p2p-release}

The proof follows the same reasoning as in the proof of Theorem~\ref{thm:kmlp-release}. For parts (1) and (2), we move to a directed metric $c'$ as in Step 1 of Algorithm~\ref{alg:p2p-rounding}, incorporate release times as in \eqref{releasetime}, solve the resulting LP \eqref{eq:lp}, and round it using Algorithm~\ref{alg:p2p-rounding}. The guarantees in parts (1) and (2) rely on the analysis in Theorems~\ref{thm:mlp-p2p} and~\ref{thm:kmlp-p2p} respectively, and the now-familiar fact that the waiting time for a tour for time $t$ is at most $t$. %in this paper and utilizes the same observations. 

For part (3), we use the following factor-3 reduction to the point-requests and use part (2) of Theorem~\ref{thm:kmlp-release}. This yields a $(41.184+\epsilon)$-approximation. 
%The obtained approximation ratio for the multi-depot $k$-MLP with point-to-point requests with release times is $3*16.994 = 50.982$, as claimed. We omit the full proof of Theorem~\ref{thm:kmlp-p2p-release} due to space constraints and instead provide below the 3-factor reduction when release times are given.
The reduction is analogous to the constant-factor reduction presented in Section~\ref{sec:multi-depot-kmlp} but proof details are
\if\conference1
 more involved and deferred to the full version.
\else
 more involved.
\fi
For path $P = r_0 s_1 d_1 \cdots s_n d_n$, we have the following relations for the latencies of the corresponding backtracking path $P_b$,
\[
\LAT^+(P_b, i) = \begin{cases} \max\left\{ 
\!\begin{aligned}%[b]
	& \LAT^+(P_b, i-1) \\
	& + c_{d_{i-1} s_{i-1}} + c_{s_{i-1} s_i}
\end{aligned} 
, T_i \right\} + c_{s_i d_i}, & i > 1 \\
0, &i = 0 \end{cases} \,.
\]

%As before, it is sufficient to prove for the single-path version.
We have the following property of the backtracking paths with release times and, hence, the approximation guarantee via the reduction. 
\if\conference0
The proof of Lemma~\ref{lem:backtrack-release-reduction} follows the same steps as that of Lemma~\ref{lem:backtrack-reduction}.
\fi

\begin{lemma} \label{lem:backtrack-release}
For any path $P$ and corresponding backtracking path $P_b$, we have
$\LAT^+(P) \leq \LAT^+(P_b) \leq 3 \LAT^+(P)$.
\end{lemma}

\begin{lemma} \label{lem:backtrack-release-reduction}
Given point-to-point requests with release times, the optimal backtracking path has total latency at most 3 times that of the optimal path. More generally, an $\alpha$-approximate backtracking path has total latency at most 3$\alpha$ times that of the optimal path. 
\end{lemma}

\if\conference0
In the remaining of the section, we prove Lemma~\ref{lem:backtrack-release}.
\begin{proof}[Proof of Lemma~\ref{lem:backtrack-release}]
Without loss of generality, we assume $P = r_0 s_1 d_1 \cdots s_n d_n$ and $P_b = r_0 s_1 d_1 s_1 \cdots s_n d_n$ from root $r_0$. 

\noindent {\em (Lower bound)}
We show by induction that $\LAT^+(P, i) \leq \LAT^+(P_b, i)$ for all $i$. Then, it would follow that
\[
\LAT^+(P) = \sum_{i=1}^n \LAT^+(P, i) \leq \sum_{i=1}^n \LAT^+(P_b, i) = \LAT^+(P_b) \,.
\]
The base case when $i=1$ holds since $\LAT^+(P, 1) = \LAT^+(P_b, 1) = \max \{c_{r_0 s_1}, T_1\} + c_{s_1 d_1}$. We now assume $\LAT^+(P, i-1) \leq \LAT^+(P_b, i-1)$ and show $\LAT^+(P, i) \leq \LAT^+(P_b, i)$. Note
\begin{align*}
\LAT^+(P, i) &= \max \{\LAT^+(P, i-1) + c_{d_{i-1} s_i}, T_i \} + c_{s_i d_i} \\
	&\leq \max \{\LAT^+(P_b, i-1) + c_{d_{i-1} s_{i-1}} + c_{s_{i-1} s_i}, T_i \} + c_{s_i d_i} \\
	&= \LAT^+(P_b, i) \,,
\end{align*}
where the inequality follows from the induction hypothesis and the triangle inequality. Hence, the induction holds and the lower bound is proved.

\noindent {\em (Upper bound)}
Similarly, we show that $\LAT^+(P_b, i) \leq 2 \LAT^+(P, i-1) + \LAT^+(P, i)$ for all $i$ by induction. Then, 
\[
\sum_{i=1}^n \left( 2 \LAT^+(P, i-1) + \LAT^+(P, i) \right) \leq 3 \sum_{i=1}^n \LAT^+(P, i)  \,,
\]
and $\LAT^+(P_b) \leq 3 \LAT^+(P)$. The base case holds since $\LAT^+(P_b, 1) = \LAT^+(P, 1)$ and $\LAT^+(P, 0) = 0$. Suppose $\LAT^+(P_b, i-1) \leq 2 \LAT^+(P, i-2) + \LAT^+(P, i-1)$. Then,
\begin{align*}
\LAT^+(P_b, i) &= \max \{ \LAT^+(P_b, i-1) + c_{d_{i-1} s_{i-1}} + c_{s_{i-1} s_i}, T_i \} + c_{s_i d_i} \\
	&\leq \max \{ 2 \LAT^+(P, i-2) + \LAT^+(P, i-1) + \\
	&\quad\quad\quad\quad c_{d_{i-1} s_{i-1}} + c_{s_{i-1} d_{i-1}} + c_{d_{i-1} s_i} , T_i \} + c_{s_i d_i} \\
	&= \max \{ 2 (\LAT^+(P, i-2) + c_{s_{i-1} d_{i-1}}) + \LAT^+(P, i-1) \\
	&\quad\quad\quad\quad + c_{d_{i-1} s_i}, T_i \} + c_{s_i d_i} \\
	&\leq 2 \left( \LAT^+(P, i-2) + c_{s_{i-1} d_{i-1}} \right)  \\
	&\quad + \max \{ \LAT^+(P, i-1) + c_{d_{i-1} s_i}, T_i \} + c_{s_i d_i} \,,
\end{align*}
where the first inequality follows from the triangle inequality and inductive hypothesis, and the second one follows from the $\max$ operator. In the last expression, note the second and last terms together are exactly $\LAT^+(P, i)$. We upper bound the first term by $2 \LAT^+(P, i-1)$:
\begin{alignat*}{1}
2 \big( \LAT^+(P, &~ i-2) + c_{s_{i-1} d_{i-1}} \big) \\
	&\leq \max \{2 \left( \LAT^+(P, i-2) + c_{s_{i-1} d_{i-1}} \right), T_{i-1} \} \\
	&\leq 2 \cdot \max \{\LAT^+(P, i-2) + c_{s_{i-1} d_{i-1}}, T_{i-1} \} \\
	&\leq 2 \left( \max \{\LAT^+(P, i-2), T_{i-1} \} + c_{s_{i-1} d_{i-1}} \right) \\
	&\leq 2 \left( \max \{\LAT^+(P, i-2) + c_{d_{i-2} s_{i-1}}, T_{i-1} \} + c_{s_{i-1} d_{i-1}} \right) \\
	&= 2 \LAT^+(P, i-1) \,,
\end{alignat*}
where the inequalities follow directly from the definition of $\max$ operator. Therefore, it follows that
\[
\LAT^+(P_b, i) \leq 2 \LAT^+(P, i-1) + \LAT^+(P, i) \,,
\]
and the induction holds. 
\end{proof}
\fi

%%% Greedy worstcase example
%\begin{figure}
%\includegraphics[width=0.45\textwidth]{worstcase.png}
%\caption{A worst-case instance for the greedy algorithm.}
%\end{figure}

\section{Experiments}
\label{sec:exp}
In this section, we will present a comparison of our algorithm and a natural greedy baseline on two public datasets, {\em viz.}, the Chicago taxi data\footnote{\url http://digital.cityofchicago.org/index.php/chicago-taxi-data-released/} and the NYC taxi data\footnote{\url http://www.nyc.gov/html/tlc/html/about/trip\_record\_data.shtml}. We focus on the single depot k-MLP problem with release times; we also evaluated our algorithms for the k-MLP problem without release times on the same datasets and obtained qualitatively similar results. 
%Before we present the details of our implementation and results, we detail the %datasets.

\subsection{Datasets}
The Chicago dataset consists of $27$ million taxi trips from $2013$ to $2016$, consisting of pick-up times, drop-off times, pick-up locations, drop-off locations, trip times and trip lengths. The pick-up and drop-off locations are specified in terms of Census Tracts instead of precise latitude and longitude values. For our analysis and graph construction, we cover the city using equal-sized geographical cells of size $1000$-sq. meters, and map the Census Tracts to these cells. Each taxi trip is then treated as a trip between two cells. We assign pairwise distances between each pair of
cells to be the driving time between  centroids of each cell, computed using Google Maps API.

The New York taxi dataset consists of more than $1$ billion taxi trips from $2009$ to $2015$ with location information. 
%%The data format is similar to the Chicago data, except that the trips have precise latitude and longitude coordinates. However, to be consistent with the Chicago dataset, in our experiments, we still 
We discretize the pick-up and drop-off locations of each trip to their respective geographical cells, and compute pairwise driving distances between all cells in New York.

We randomly choose one weekday for both datasets, 7/17/2015 for Chicago and 5/31/2013 for New York, and consider all taxi rides in a one-hour (5PM to 6PM) and two-hour (5PM to 7PM) window. 
%For the New York data, we filtered out all trips of length less than $5$km.
Also, since the datasets do not contain release-times (when the trip-request arrived), we set the release-times to be equal to the pick-up times of the requests. We also choose a random cell in each city as the depot location.
\if\conference1
Note the distribution of taxi rides in Chicago is relatively more spread out than that in New York. See the full paper version for a visualization.
\else

Figure~\ref{fig:demands} plots a random sample of 100 taxi pick-up (white) and drop-off (red) locations in the 1-hour window for Chicago and New York. As seen from the figure, the distribution of taxi rides in Chicago is relatively more spread out. 
%in New York and Chicago are quite different. The pickup and drop-off points in Chicago are relatively more spread out in the city, compared to New York where the taxi rides are more dense and concentrated in a single area (Manhattan).

\begin{figure}[ht]
\centering
\subfigure[New York City]{
\includegraphics[scale=0.15]{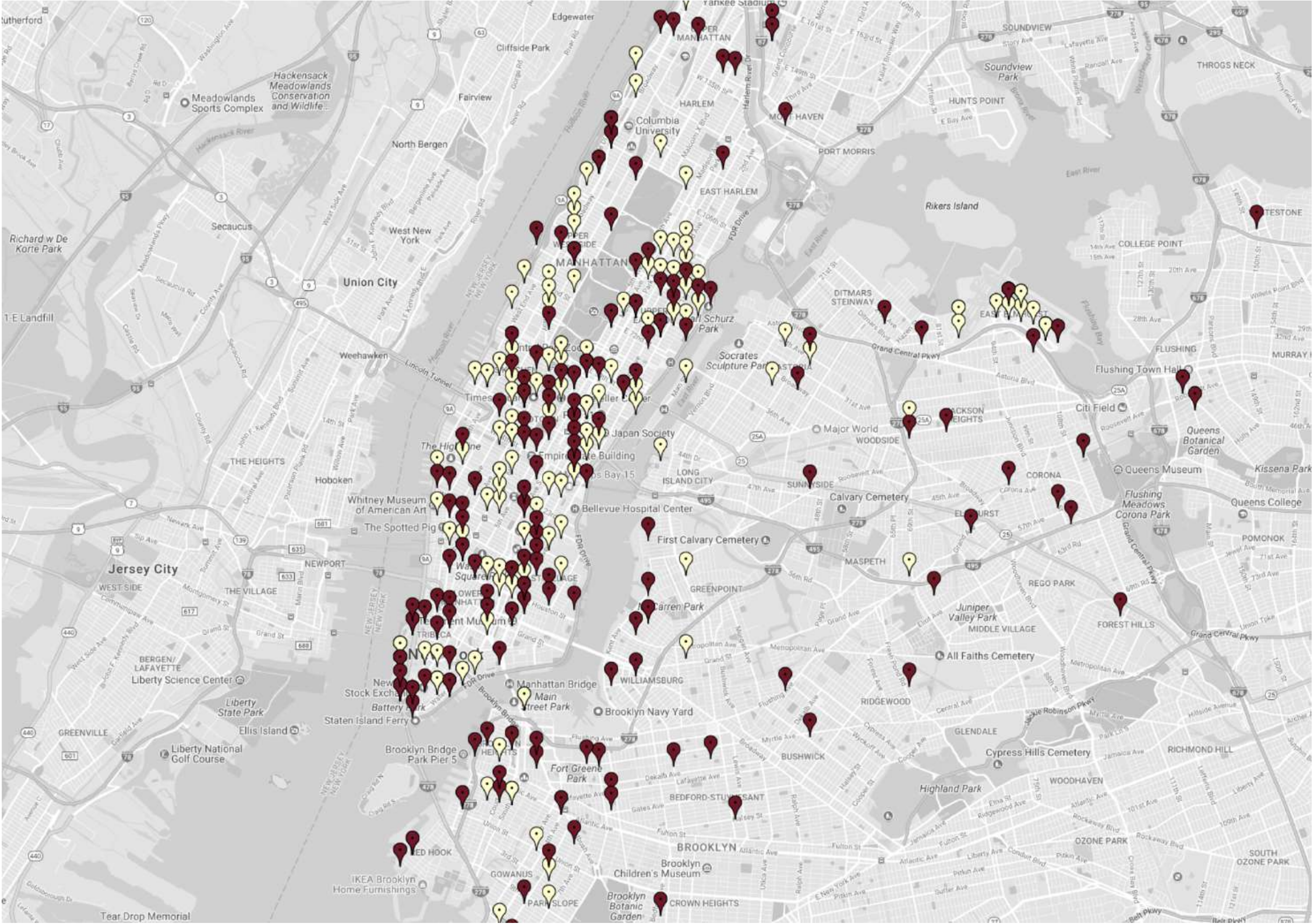}
\label{fig:nyc-demands}
}%
\subfigure[Chicago]{
\includegraphics[scale=0.15]{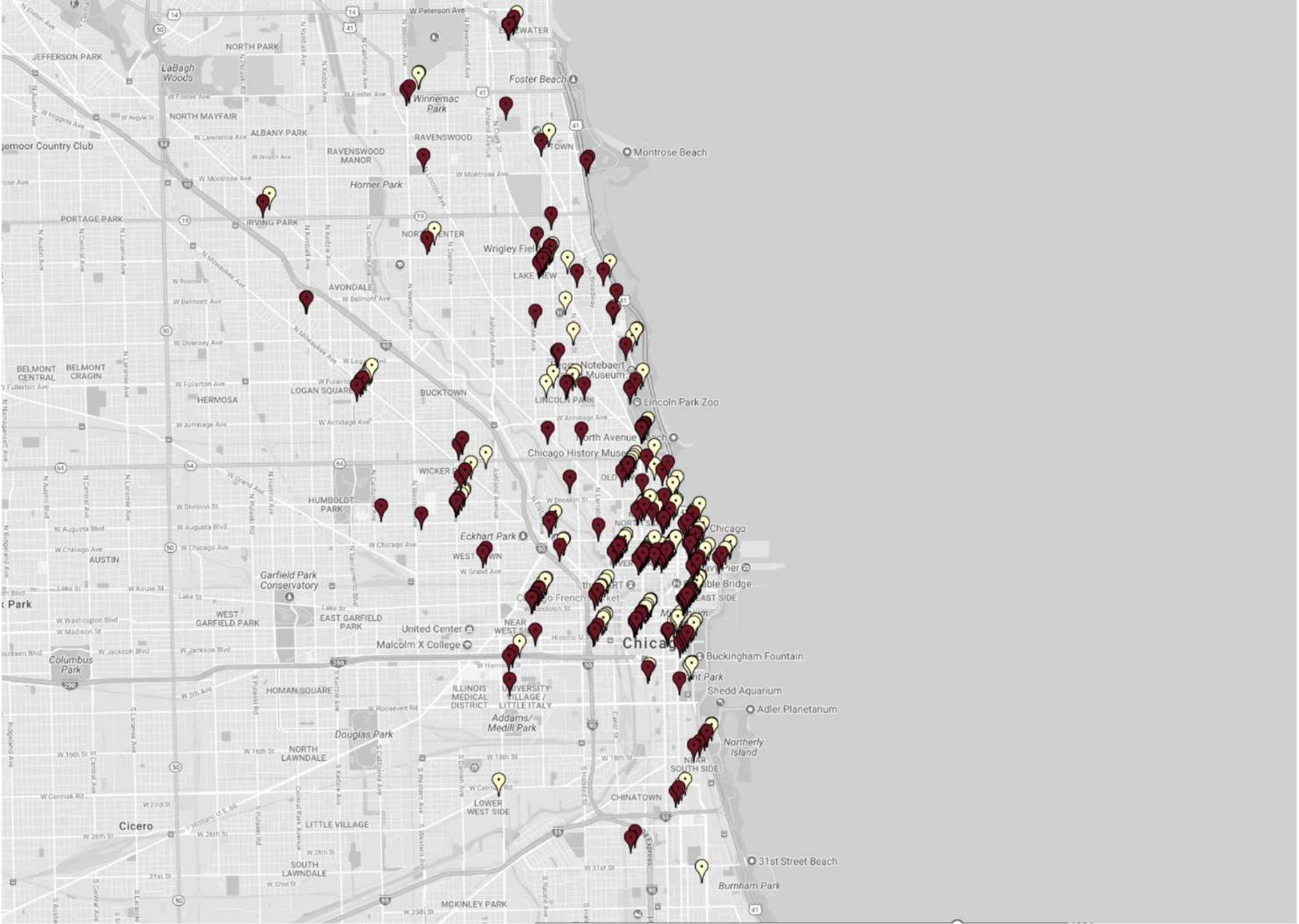}
\label{fig:chi-demands}
}
\caption{Spatial distributions of random samples of 100 ride requests in New York City and Chicago.}
\label{fig:demands}
\end{figure}
\fi

\subsection{Algorithm Implementations}
We first describe the natural greedy algorithm for the problem, that we use as a baseline. For each taxi that becomes idle, the greedy algorithm scans through all unassigned trip requests, and assigns the trip request with the smallest resulting finish-time. The resulting solution has at most $k$ paths branching out from the root and we refer to such solution as a {\em $k$-path solution}. A path of a $k$-path solution will correspond to the path of a taxi.

Even though our Algorithm~\ref{alg:p2p-rounding} from Section~\ref{sec:single-depot-kmlp} has constant approximation guarantees, it is relatively complex and hard to implement in practice. We therefore use the following  practical heuristic (that we call kMLP-Fast) based on Algorithm~\ref{alg:p2p-rounding} (in the sequel, $n$ refers to the number of trip requests and $k$ refers to the number of taxis): 
\begin{enumerate}
\item {Sort the requests in increasing order of their release times}
\item {Define a layer($i$) solution as the greedy $k$-path solution considering only the first $2^i$ requests. Generate layers for $i \in [1,2, \ldots,\log(n)]$ where each layer $i$ consists of at most $k$ paths of the layer($i$) solution.}
\item {Create a ``concatenation graph'' %(as described previously) 
as follows: Each node in the concatenation graph corresponds to a path in a given layer. We create an edge from path $P_i$ in layer $i$ to path $P_j$ in layer $j$  for $j > i$, and set the
        edge cost as follows:
        $e(P_i, P_j) = \sum_{t \in P_j} (\LAT(P_i + P_j, t) - \LAT(P_j, t))  + (n  - 2^j)\cdot\text{length}(P_j)$, where $t$ indexes the requests in $P_j$, $P_i + P_j$ represents the path obtained by appending $P_j$ to $P_i$ and removing duplicates, and $\text{length}(P_j)$ is the traversal time of the path (which includes both waiting and driving times).
Essentially, the edge cost is an upper bound on the additional latency incurred by trips in $P_j$ and all subsequent trips in higher layers, when appending $P_j$ to $P_i$. 
        Create a dummy start node $s$ corresponding to an empty path, and connect every node to $s$ with edge costs computed as before. Create a finish node $t$ and connect every path in the last layer ($i=\log(n)$) to $t$, with edge cost $0$.}
\item {Assign a capacity of $1$ to each edge in the concatenation graph, and run a min-cost flow algorithm  from $s$ to $t$ with the target flow amount $k$. Concatenate the paths traversed by each flow (removing duplicates) to generate the final solution.}
\end{enumerate}

The main simplification that the kMLP-Fast algorithm performs (compared to Algorithm~\ref{alg:p2p-rounding}) is that it greedily defines requests for each layer by sorting them by their release times. 
For each layer, it then creates a $k$-path solution using the greedy-heuristic as a subroutine to cover all  nodes in that layer. This differs from the paths obtained in Algorithm~\ref{alg:p2p-rounding}.
%(In the theoretical analysis, the $k$-paths are obtained by breaking up the trees obtained from the LP solution, or more combinatorially, via arborescence packing.) 
Finally, the algorithm concatenates the paths across different layers using a min-cost flow algorithm (as opposed to a shortest path in the concatenation-graph, since we now work with at most $k$ paths per layer.)

\subsection{Latency Objective}
We compare the kMLP-Fast and greedy algorithms on the total-latency objective, which includes both the waiting time and the driving time for a request. Figure~\ref{fig:comp} illustrates the performance of \nolinebreak \mbox{the algorithm on both datasets.} Note that on the y-axis, we report values representing the {\em sum} of the corresponding metric over all requests in the observed interval. 

\begin{figure}[b]
\vspace*{-5ex}
\centering
\subfigure[New York City]{
\includegraphics[scale=0.19]{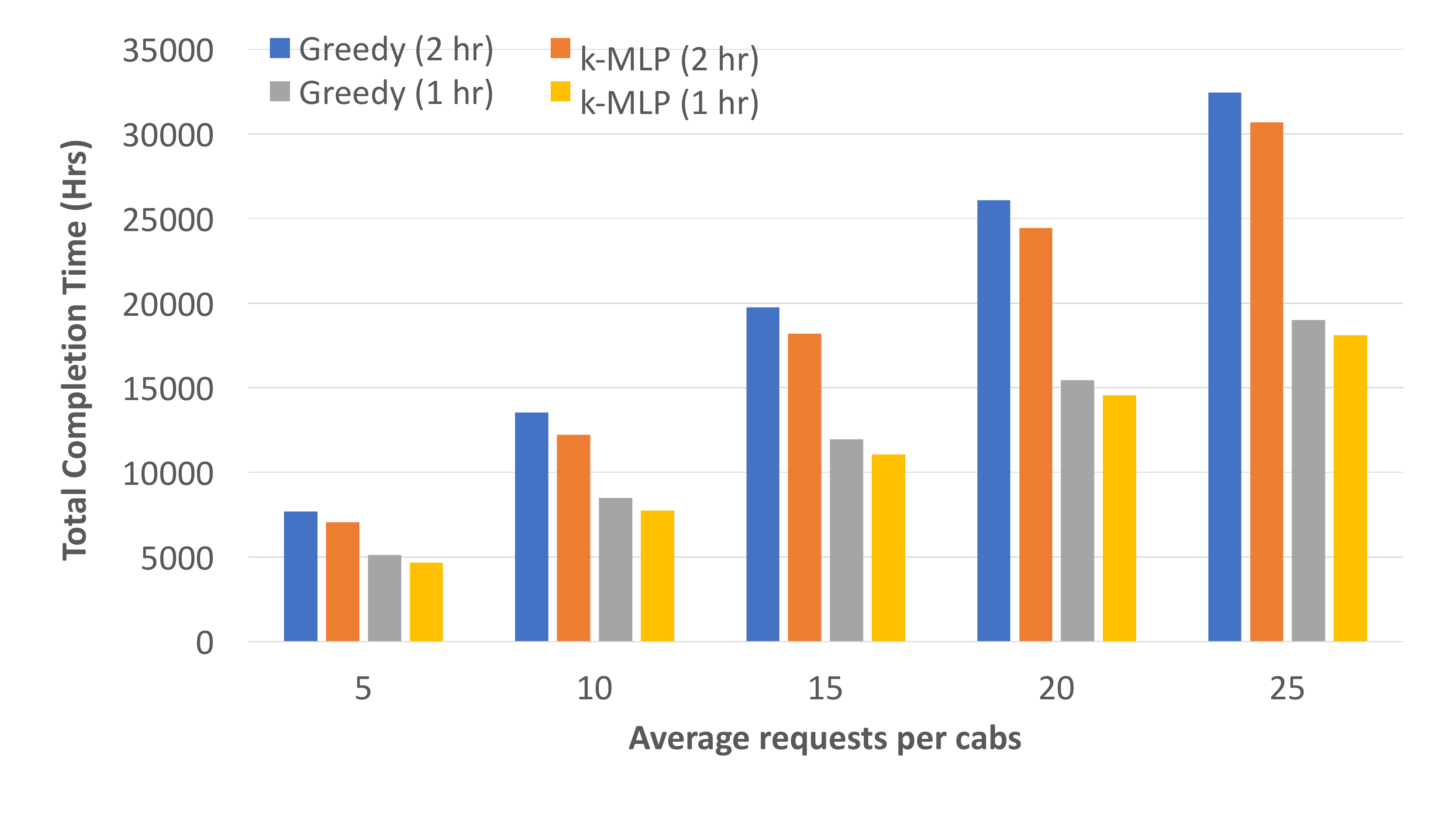}
\label{fig:nyc-comp}
}
\quad
\subfigure[Chicago]{
\includegraphics[scale=0.19]{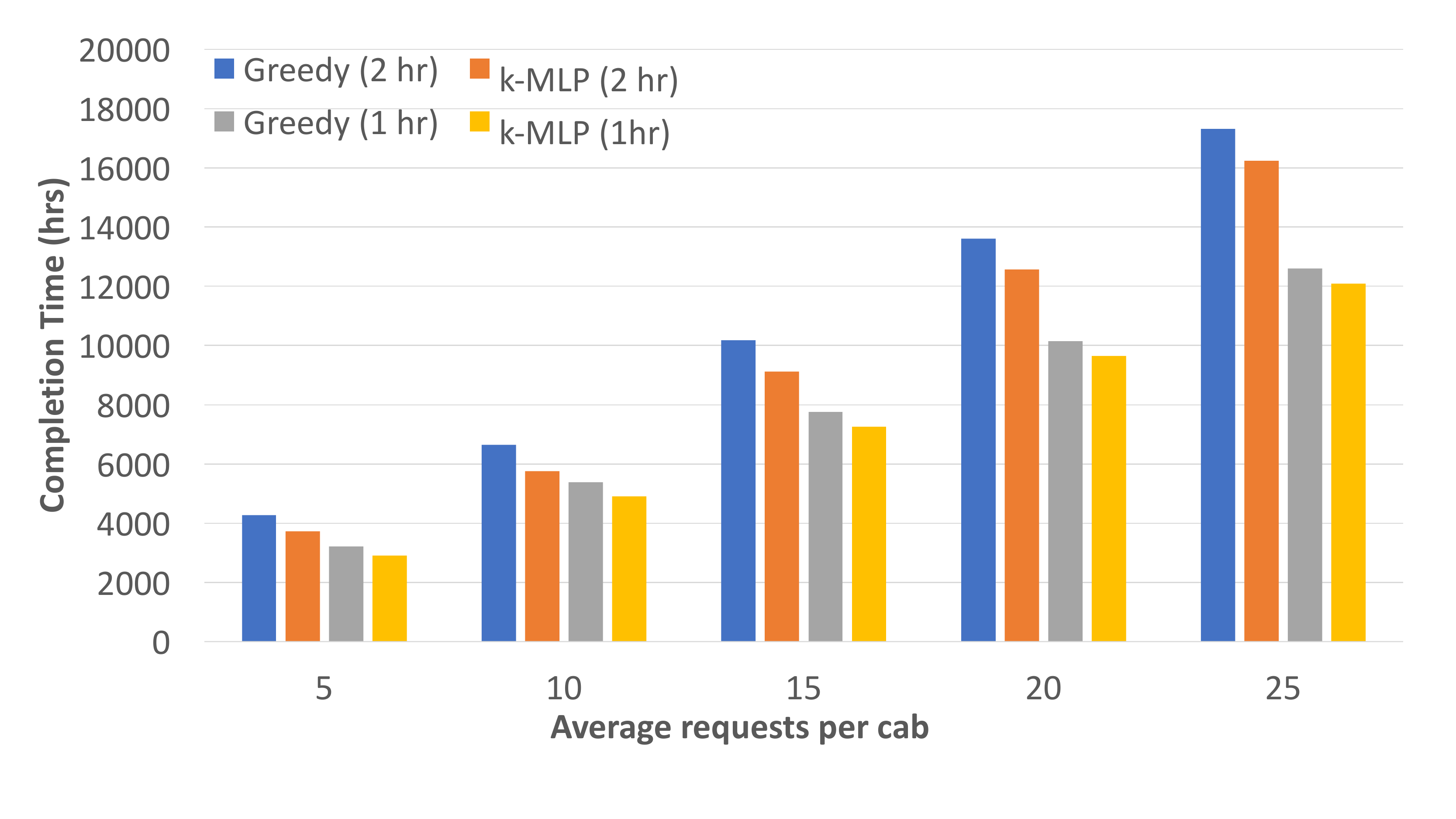}
\label{fig:chi-comp}
}
\vspace*{-2ex}
\caption{Total latency of all trips in hours}
\label{fig:comp}
\end{figure}

\if\conference1
\begin{figure}[ht]
\centering
\subfigure[New York City - total length in hours]{
\includegraphics[scale=0.19]{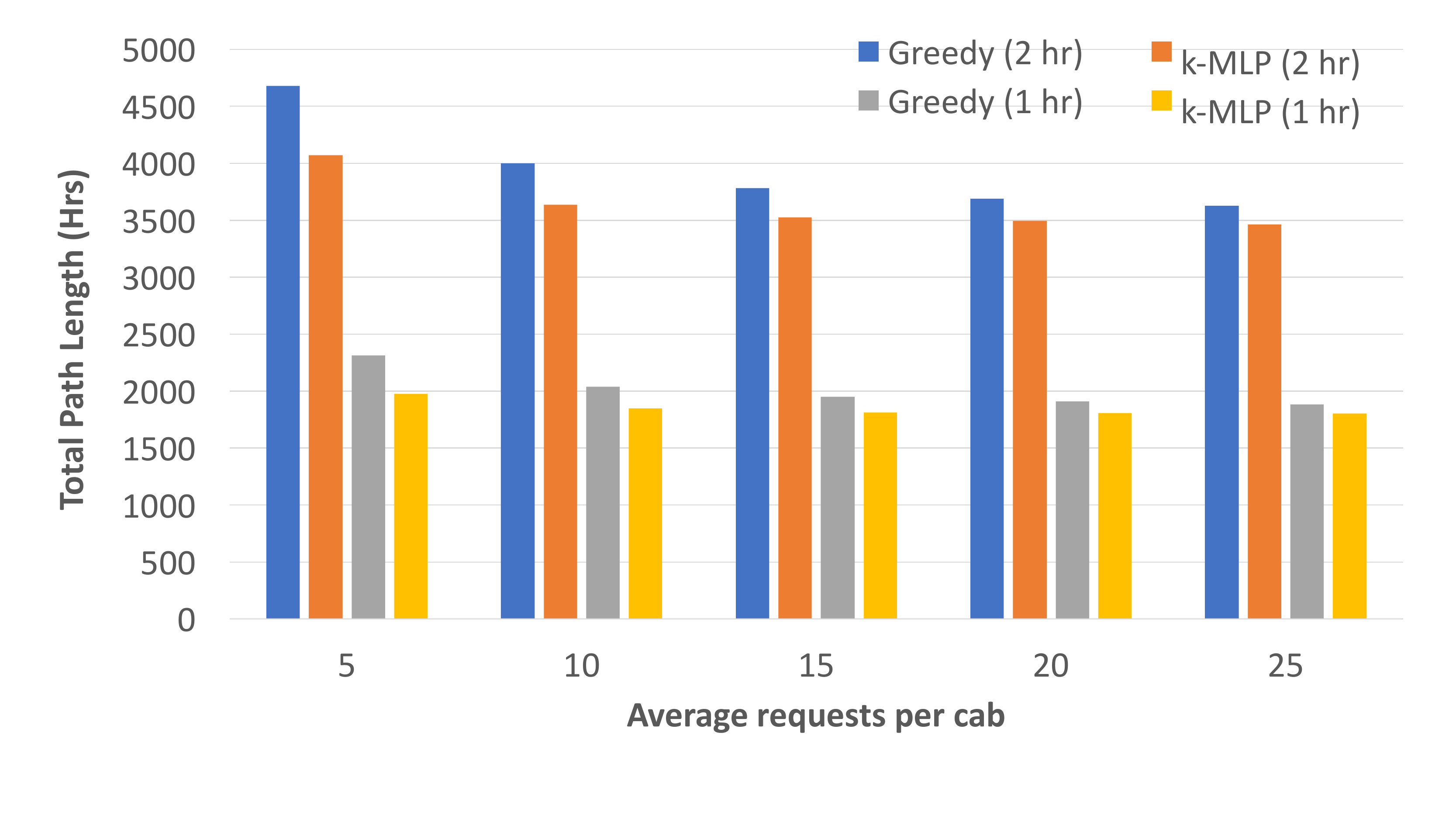}
\label{fig:nyc-len}
}
\quad
\subfigure[New York City - total idle time of cabs in hours]{
\includegraphics[scale=0.19]{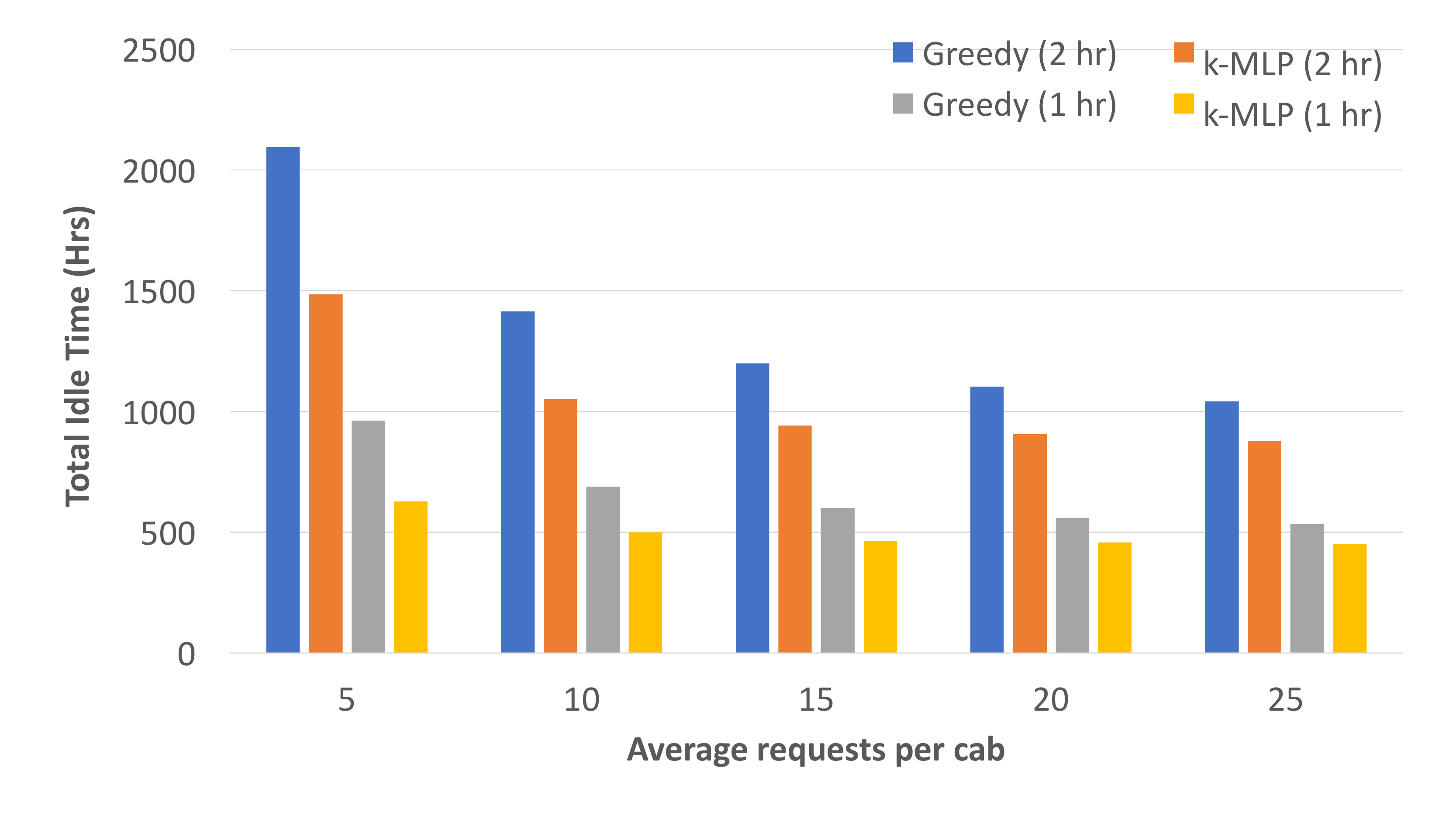}
\label{fig:nyc-len}
}
\quad
\subfigure[New York City - balanced allocations]{
\includegraphics[scale=0.19]{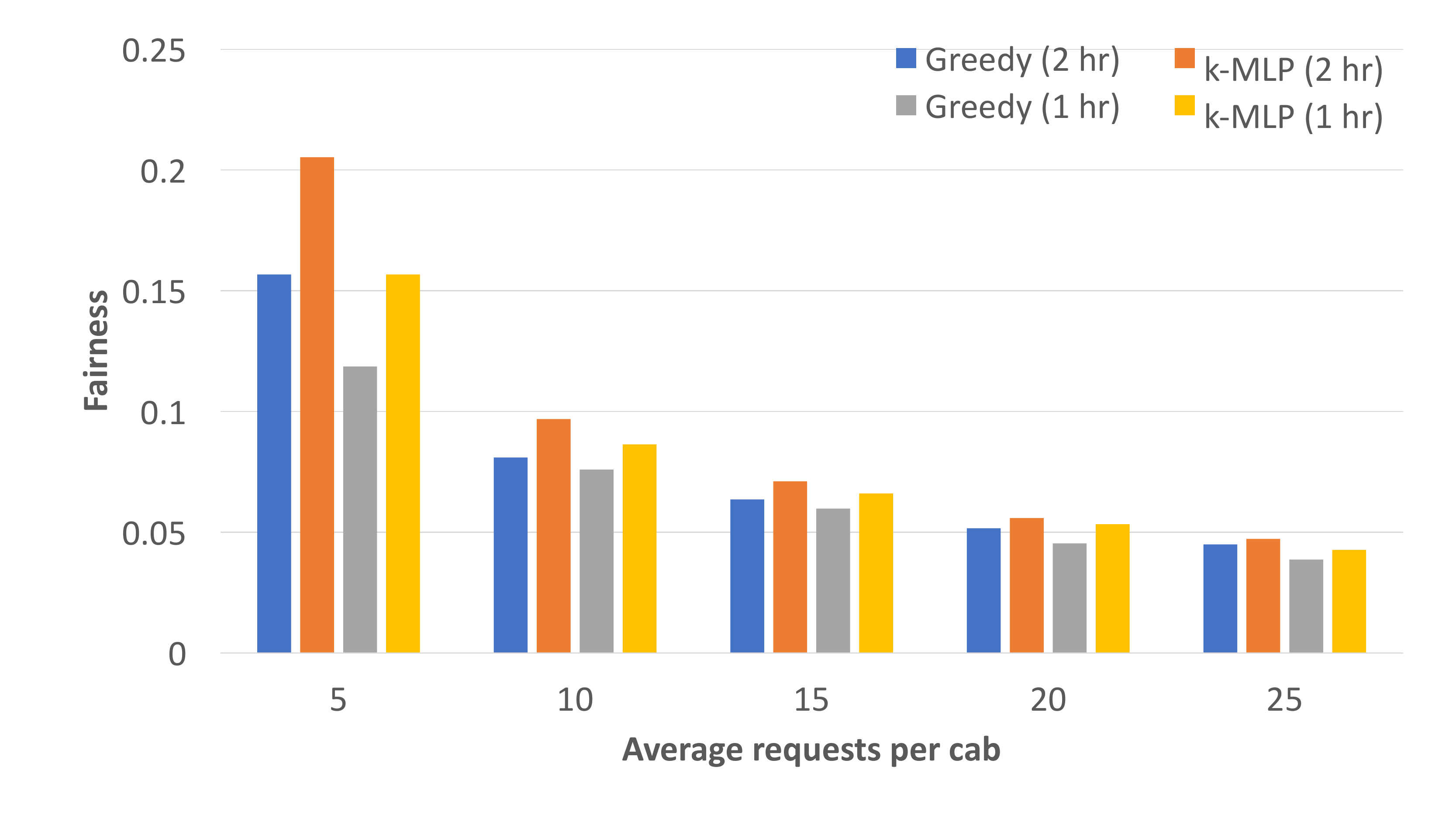}
\label{fig:nyc-len}
}
\caption{The performance of kMLP-Fast algorithm w.r.t. other natural measures.}
\label{fig:des}
\end{figure}
\else
\begin{figure*}[ht]
\centering
\subfigure[New York City - total length in hours]{
\includegraphics[scale=0.20]{nyc-length.pdf}
\label{fig:nyc-len}
}
\quad
\subfigure[Chicago - total length in hours]{
\includegraphics[scale=0.20]{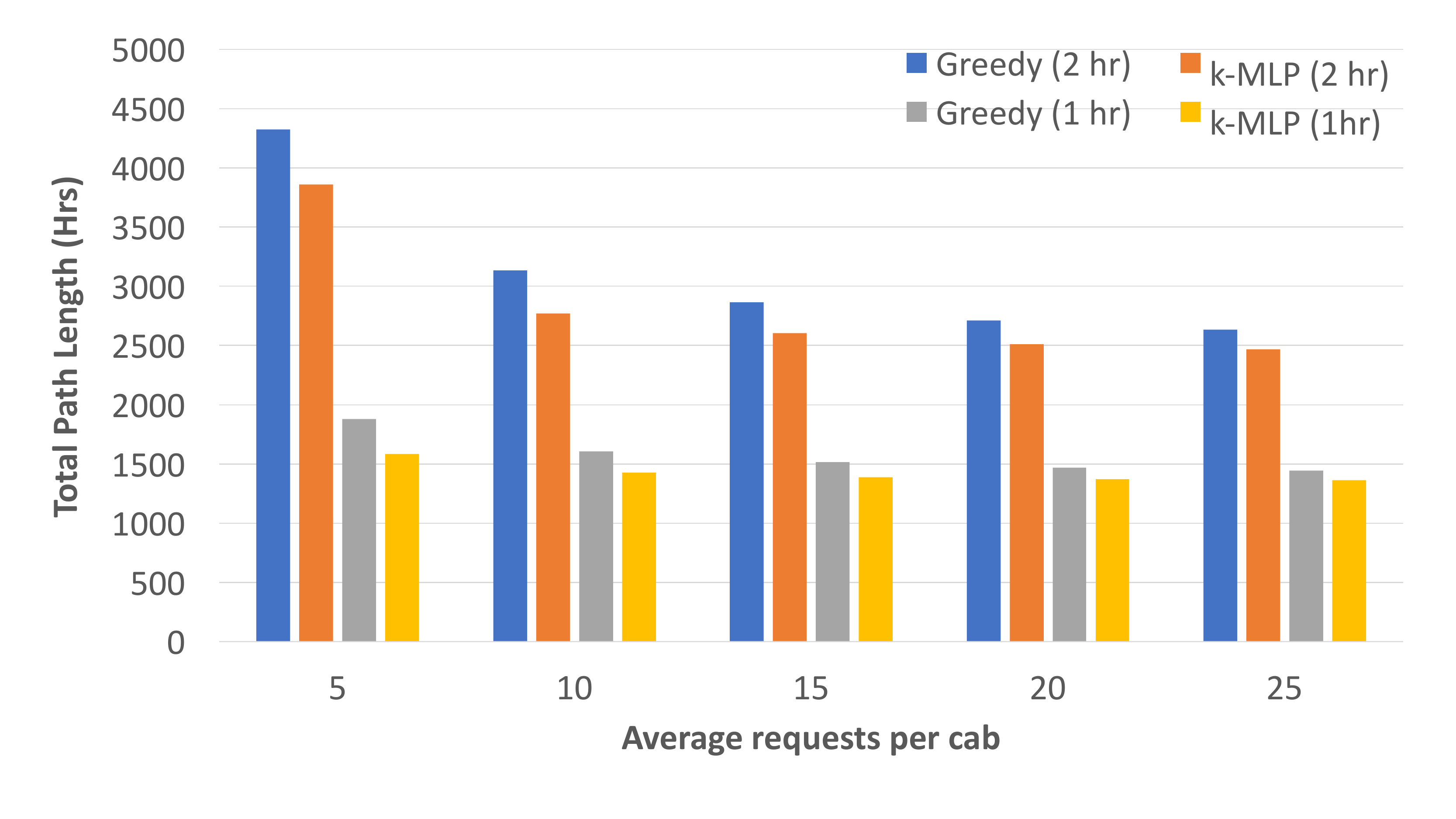}
\label{fig:chi-len}
}
\quad
\subfigure[New York City - total idle time of cabs in hours]{
\includegraphics[scale=0.20]{nyc-idle.pdf}
\label{fig:nyc-len}
}
\quad
\subfigure[Chicago - total idle time of cabs in hours]{
\includegraphics[scale=0.20]{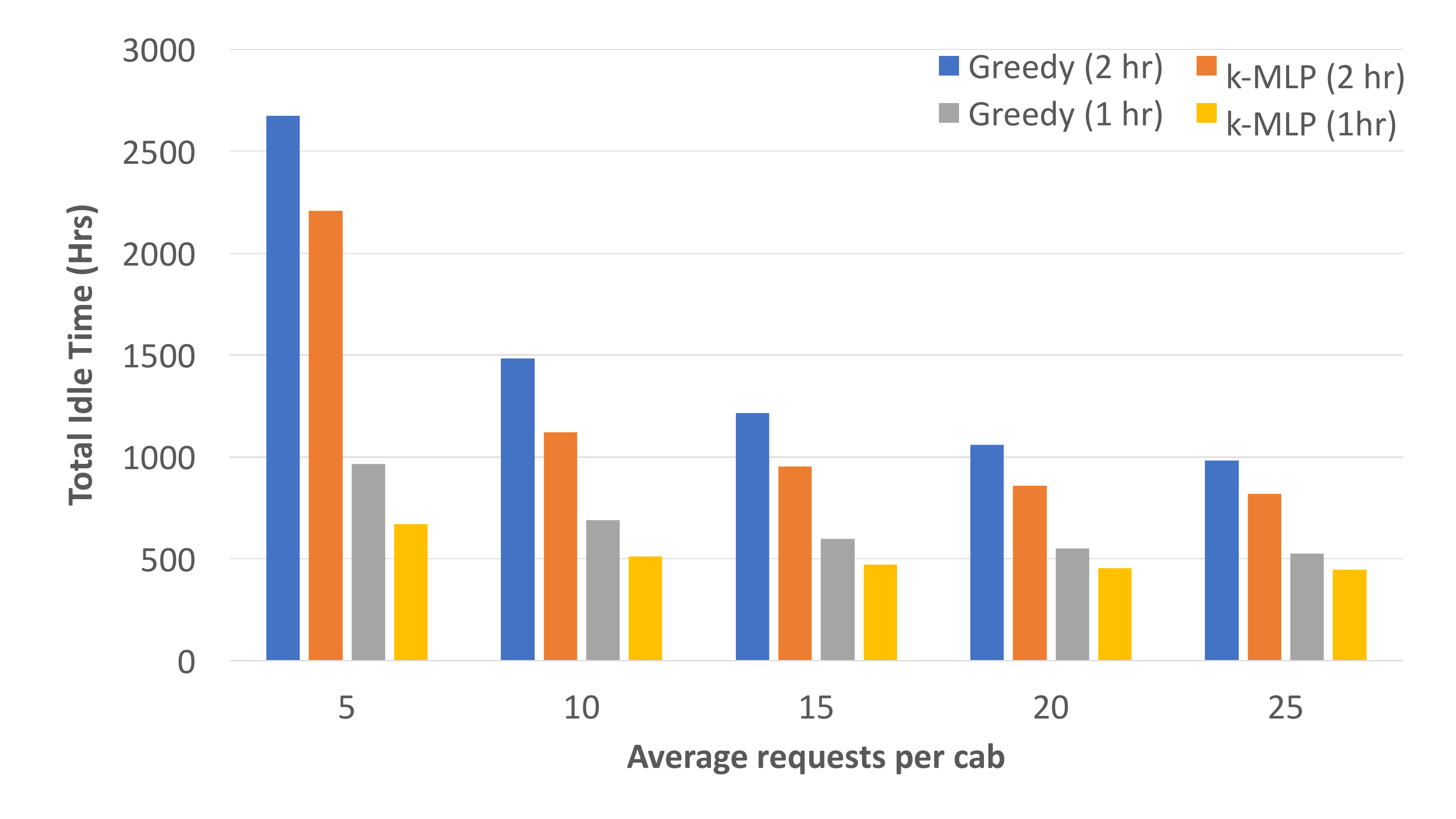}
\label{fig:chi-len}
}
\quad
\subfigure[New York City - balanced allocations]{
\includegraphics[scale=0.20]{nyc-fair.pdf}
\label{fig:nyc-len}
}
\quad
\subfigure[Chicago - balanced allocations]{
\includegraphics[scale=0.20]{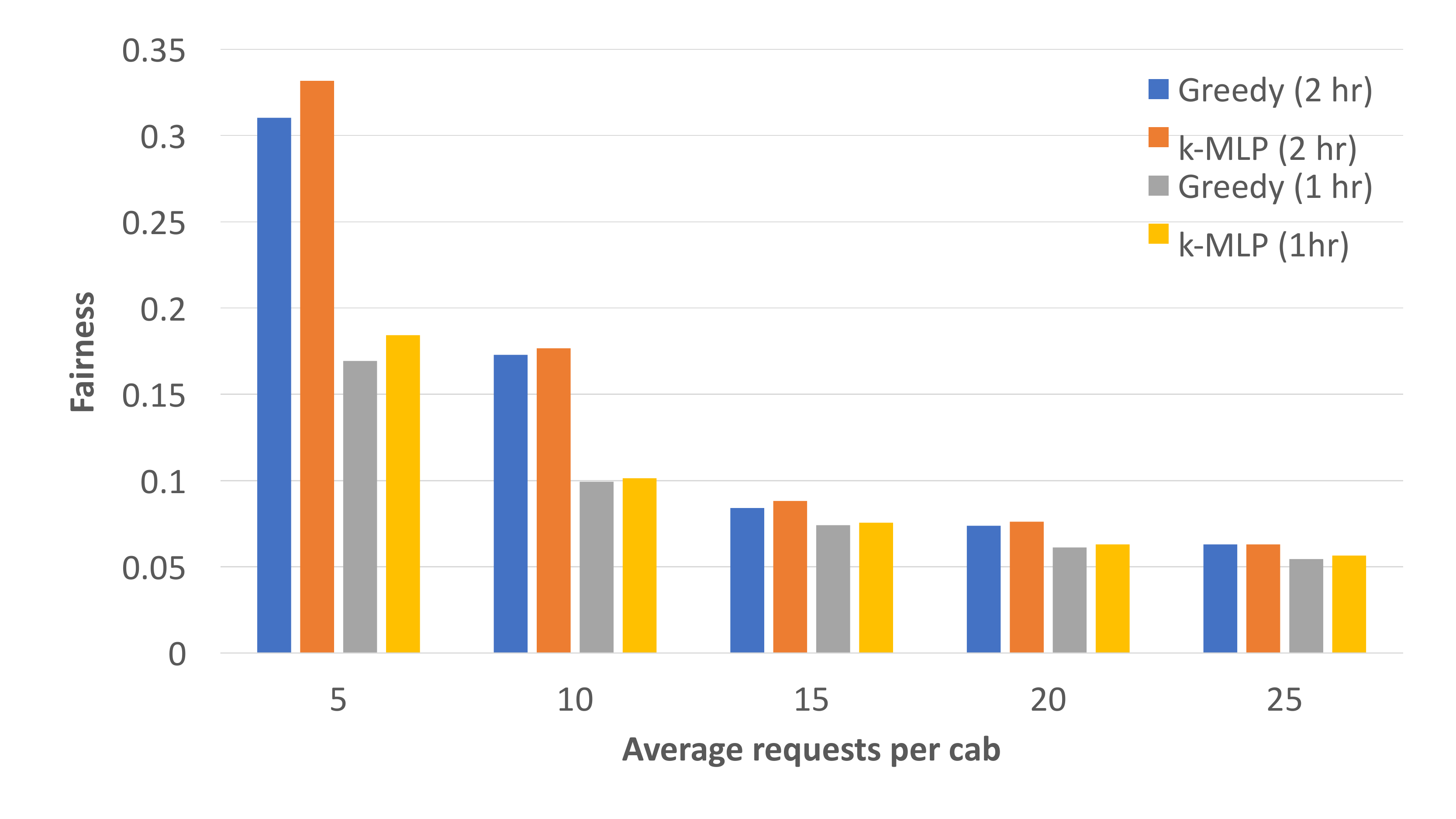}
\label{fig:chi-len}
}
\caption{The performance of kMLP-Fast algorithm w.r.t. other natural measures.}
\label{fig:des}
\end{figure*}
\fi

One trend we observe in both cities is that as the number of taxis increases, our algorithm outperforms the greedy algorithm by a wider margin. Indeed, in New York City, our algorithm does better by around 5\% - 8\% as the average number of trips for a taxi decreases from $25$ to $5$. A similar trend, with the difference ranging from 6\% to 13\%, is observed in Chicago. The reason for the higher difference in Chicago compared to New York City might be attributed to
\if\conference0
 the different demand distributions in the two cities (see Figure \ref{fig:demands}).
\else
 the different demand distributions in the two cities.
\fi
The requests in New York are more densely concentrated in a smaller region, which helps the greedy algorithm.

\subsection{Other Desiderata}
While the user satisfaction in terms of latency is a good characteristic of a ride sharing service, other characteristics such as the total distance covered by the taxis in the system, the efficiency of the taxi in terms of taxi idle times, the load on each taxi, etc are also good properties of any ride sharing service. We measured the performance of the kMLP-Fast algorithm w.r.t. these metrics. 
\if\conference1
Figure~\ref{fig:des} illustrates the relative performance of our algorithm compared to the greedy baseline for the New York City dataset. We also obtain similar performances for the Chicago dataset and provide these results in the full paper version.
\else
Figure~\ref{fig:des} illustrates the relative performance of our algorithm compared to the greedy baseline.
\fi

Even though the kMLP-Fast algorithm 
%Algorithm~\ref{alg:p2p-rounding} (and hence the kMLP-Fast algorithm) is designed to 
is optimized for the latency objective, it performs surprisingly well in minimizing the total distance a taxi travels to serve its rides.  In fact, the cabs travel between 5\% and 15\% less on this measure. Another promising result is the performance on the measure of idle time a taxi spends transitioning between rides. Here again, our algorithm performs significantly better by 15\% to 35\% over the greedy baseline. Finally, we measured the fairness as the coefficient of variation of the total trip lengths of taxis.  A fair allocation ensures taxis are generally equally loaded, resulting in a smaller value of this measure. On this measure, the greedy algorithm under-performs in New York City while the difference is much less in Chicago. Note that the greedy algorithm tends to produce more balanced paths by its design.
\eject
\bibliographystyle{ACM-Reference-Format}
\bibliography{refs}

\end{document}